\documentclass{fundam}

\publyear{22}
\papernumber{2102}
\volume{185}
\issue{1}

\usepackage{url} 
\usepackage[ruled,lined,linesnumbered]{algorithm2e}
\usepackage{graphicx}
\usepackage{tikz}
\usetikzlibrary{automata, positioning, arrows}

\usepackage[capitalize,noabbrev]{cleveref}
\usepackage{amsmath,amssymb,amsfonts}
\usepackage{mathabx}
\usepackage[inline]{enumitem}
\usepackage{stmaryrd}

\newcommand{\preset}[1]{\ensuremath{{\bullet #1}}}
\newcommand{\postset}[1]{\ensuremath{{#1\bullet}}}

\newcommand{\enabled}[2]{\ensuremath{#1[#2\rangle}}
\newcommand{\step}[4]{\ensuremath{\enabled{(#1,#2)}{#3}(#1,#4)}}

\newcommand{\Paths}[2]{\ensuremath{\textnormal{\textit{Paths}}(#1,#2)}}

\newcommand{\concurrent}[2]{\ensuremath{#1\parallel#2}}
\newcommand{\notconcurrent}[2]{\ensuremath{#1\nparallel#2}}
\newcommand{\concurrentSet}[1]{\ensuremath{\parallel\!\!(#1)}}

\newcommand*\conj[1]{\overline{#1}}
\newcommand{\RunNets}[2]{\ensuremath{\Pi\ifx\\#1\\%
\else
(#1) %
\fi}}
\newcommand{\idivpoint}[2]{\ensuremath{#1 \llcurly #2}}
\newcommand{\idivpoints}[1]{\ensuremath{\Delta(#1)}}
\newcommand{\divpoints}[1]{\ensuremath{\Delta(#1)}}
\newcommand{\divinfo}[2]{\ensuremath{\mathcal{R}_{#2}\ifx\\#1\\%
\else
(#1)%
\fi}}
\newcommand{\divinfoh}[2]{\ensuremath{\hat{\mathcal{R}}_{#2}\ifx\\#1\\%
\else
(#1)%
\fi}}
\newcommand{\divinfoS}[2]{\ensuremath{\hat{\mathcal{R}}_{#2}\ifx\\#1\\%
\else
(#1)%
\fi}}
\newcommand{\reaches}[2]{\ensuremath{\,\rightrightarrows_{#2}\ifx\\#1\\
\else
\!\!(#1)%
\fi}}

\newcommand{\pathset}[1]{\ensuremath{[#1]}}
\newcommand{\markingset}[1]{\ensuremath{\langle#1\rangle}}

\newcommand{\Bag}[1]{\ensuremath{\mathcal{B}(#1)}}
\newcommand{\mcup}{\ensuremath{\uplus}}
\newcommand{\enabledTrans}[3]{\ensuremath{(#1,#2)[#3\rangle}}
\newcommand{\allEnabled}[2]{\ensuremath{\textit{en}(#1,#2)}}

\newcommand{\reachableM}[2]{\ensuremath{R(#1,#2)}}

\newcommand{\proofend}{}
\newcommand{\defend}{$\,$ \hfill ${\lrcorner}$}

\newcommand{\proofsize}{}

\newcommand*{\FW}{{\small FC}-{\small WF}-net}
\newcommand*{\FWs}{\FW{}s}
\newcommand*{\AFW}{{\small A}\FW}
\newcommand*{\AFWs}{{\small A}\FWs}

\makeatletter
\newcommand*{\coloneqq}{\mathrel{\vcenter{\baselineskip0.5ex \lineskiplimit0pt
                     \hbox{\scriptsize.}\hbox{\scriptsize.}}}%
                     =}
\makeatother
\newcommand{\Set}[1]{\ensuremath{\{\,#1\,\}}}
\newcommand{\multiset}[1]{\ensuremath{[\,#1\,]}}
\newcommand{\given}{\ensuremath{\mid}}

\newtheorem{theoremit}{Theorem}[section]
\newtheorem{lemmait}{Lemma}[section]
\newtheorem{corollaryit}{Corollary}[section]

\usepackage{cancel}

\newcommand{\domi}[2]{\ensuremath{\ifx\\#1\\%
\else
#1 \in %
\fi\underline{\mathit{dom}}(#2)}}
\newcommand{\pdomi}[2]{\ensuremath{\ifx\\#1\\%
\else
#1 \in %
\fi\underline{\mathit{pdom}}(#2)}}
\newcommand{\sdomi}[2]{\ensuremath{\ifx\\#1\\%
\else
#1 \in %
\fi\mathit{dom}(#2)}}
\newcommand{\notspdomi}[2]{\ensuremath{\ifx\\#1\\%
\else
#1 \in %
\fi\cancel{\mathit{pdom}}(#2)}}
\newcommand{\spdomi}[2]{\ensuremath{\ifx\\#1\\%
\else
#1 \in %
\fi\mathit{pdom}(#2)}}
\newcommand{\isdomi}[2]{\ensuremath{\ifx\\#1\\%
\else
#1 \in %
\fi\mathit{idom}}(#2)}
\newcommand{\ispdomi}[2]{\ensuremath{\ifx\\#1\\%
\else
#1 = %
\fi\mathit{ipdom}\ifx\\#2\\%
\else
(#2)%
\fi}}

\newcommand{\pdomfront}[1]{\ensuremath{\mathit{PDF}\ifx\\#1\\%
\else
(#1)%
\fi}}
\newcommand{\itpdomfront}[1]{\ensuremath{\mathit{PDF}^{+}(#1)}}

\begin{document}

\title{Deciding Reachability and the Covering Problem with Diagnostics \\for Sound Acyclic Free-Choice Workflow Nets}

\address{T. Prinz, Course Evaluation Service, Friedrich Schiller University Jena}

\author{Thomas M. Prinz\\
Course Evaluation Service\\
Friedrich Schiller University Jena \\ Am Steiger 3, Haus 1, 07743 Jena, Germany\\
Thomas.Prinz{@}uni-jena.de
\and Christopher T. Schwanen\\
Chair of Process and Data Science (PADS)\\
RWTH Aachen University \\ Ahornstra{\ss}e 55, 52074 Aachen, Germany \\
schwanen{@}pads.rwth-aachen.de
\and Wil M.\,P. van der Aalst\\
Chair of Process and Data Science (PADS)\\
RWTH Aachen University \\ Ahornstra{\ss}e 55, 52074 Aachen, Germany \\
wvdaalst{@}pads.rwth-aachen.de
}

\maketitle

\runninghead{T.M. Prinz, C.T. Schwanen, W.M.P. van der Aalst}{Deciding Reachability with Diagnostics}

\begin{abstract}
A central decision problem in Petri net theory is \emph{reachability} asking whether a given marking can be reached from the initial marking. Related is the \emph{covering problem} (or sub-marking reachbility), which decides whether there is a reachable marking covering at least the tokens in the given marking. For live and bounded free-choice nets as well as for sound free-choice workflow nets, both problems are polynomial in their computational complexity. This paper refines this complexity for the class of sound \emph{acyclic} free-choice workflow nets to a quadratic polynomial, more specifically to $O(P^2 + T^2)$. Furthermore, this paper shows the feasibility of accurately explaining why a given marking is or is not reachable. This can be achieved by three new concepts: admissibility, maximum admissibility, and diverging transitions. \emph{Admissibility} requires that all places in a given marking are pairwise concurrent. \emph{Maximum admissibility} states that adding a marked place to an admissible marking would make it inadmissible. A \emph{diverging transition} is a transition which originally ``produces'' the concurrent tokens that lead to a given marking. In this paper, we provide algorithms for all these concepts and explain their computation in detail by basing them on the concepts of concurrency and post-dominance frontiers --- a well known concept from compiler construction. In doing this, we present straight-forward implementations for solving (sub-marking) reachability.
\end{abstract}


\section{Introduction}
\label{sec:Introduction}

\emph{Reachability} is a central problem in Petri net theory deciding if a given marking can be reached from the initial marking. This decision is crucial for showing whether certain desired or undesired properties of a system under investigation are fulfilled, or not. Liveness, boundedness, and safeness (as global properties) as well as the absence of deadlocks, livelocks, and undesired states during conformance checking (as local properties) are examples of such properties. Following from such properties, it is sometimes also important to know \emph{why} a given marking is or is not reachable. The ability to explain the decision and to effectively compute this decision are important as reachability is at the core of many verification approaches. The complexity class of the general reachability problem for Petri nets is \emph{Ackermann-complete} \cite{DBLP:journals/jacm/CzerwinskiLLLM21,CzerwinskiO2022,Leroux2022}. For general safe nets, the complexity is reduced to \verb+PSPACE+ \cite{DBLP:journals/tcs/ChengEP95}. Esparza~\cite{DBLP:journals/tcs/Esparza98} stated that the reachability problem in safe and live free-choice nets is NP-complete as it can be reduced to the \emph{CNF-SAT} problem. For \emph{cyclic} free-choice nets, a restricted subclass of live and bounded free-choice nets where the initial marking is a home marking, reachability can be decided in polynomial time \cite{DBLP:journals/tcs/DeselE93}. Eventually, Yamaguchi \cite{DBLP:journals/ieicet/Yamaguchi14} shows a polynomial time complexity for reachability in sound extended free-choice workflow nets, but the exact polynomial is unknown.

In this paper, we extend the work presented in \cite{DBLP:conf/apn/PrinzSA25}, which showed first that reachability for \emph{sound acyclic (simple) free-choice workflow nets} can be solved in quadratic time, $O(P^2 + T^2)$, and that \emph{sub-marking reachability} \cite{DBLP:journals/corr/abs-2411-01592} (the \emph{covering problem}, i.\,e., if a partial marking is reachable) can be solved in the same computational complexity. This paper extends the approach by modified and additional algorithms, which provide diagnostics on why a (sub-)marking is reachable or not. For this reason, the pure decision problem is extended to be explainable. To the best knowledge of the authors, this is the first attempt to algorithmically explain (non-)reachability of a given marking.

\Cref{fig:example} shows an example of a sound acyclic free-choice workflow net. System analysts could ask whether a marking with places $p9$, $p12$, and $p16$ having tokens is reachable from the initial marking, or not. This paper will show that this marking is reachable as a sub-marking since (1) all places are pairwise concurrent (i.\,e., they can have tokens at the same moment) and (2) there are transitions $t1$ and $t8$ causing that all places can have tokens in the same marking.

\begin{figure}[tb]
	\centering
		\includegraphics[width=1.00\textwidth]{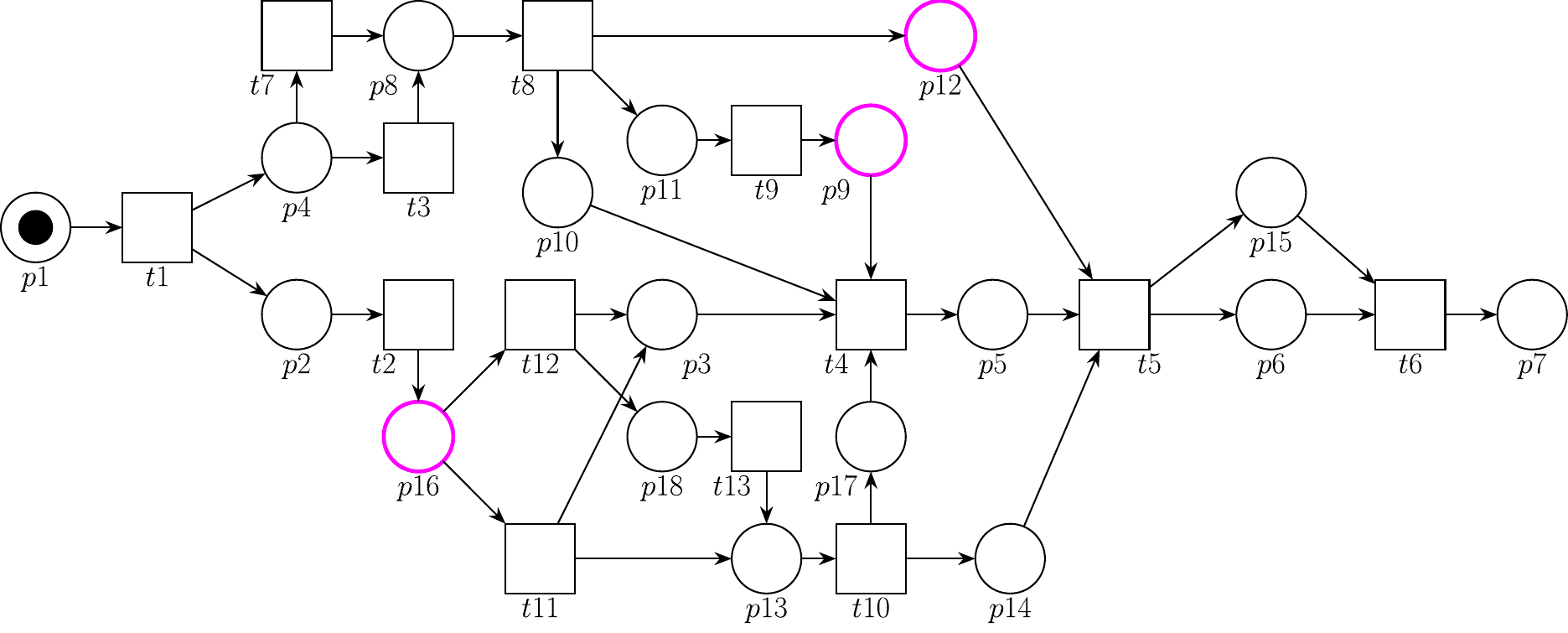}
		\vspace{-0.35cm}
	\caption{A sound acyclic free-choice workflow net focusing on a marking with tokens on $p9$, $p12$, and $p16$ (colored in pink).}
	\label{fig:example}
\end{figure}

Answering reachability questions usually focuses on \emph{concrete occurrence sequences} (traces) to the desired marking. This requires a kind of state space exploration resulting in a trace. Such a trace can be used (e.\,g., in a simulation) to argue why a marking is reachable. However, this trace-based approach usually fails in explaining why a marking is \emph{not} reachable (since there is no trace for this case). Other approaches for showing reachability provide ``just'' a simple decision as diagnostics \cite{DBLP:journals/tcs/DeselE93,DBLP:journals/ieicet/Yamaguchi14}. 

The here presented algorithm demonstrates that most decisions on whether a marking is reachable can be decided by \emph{concurrency}. Thereby, we introduce the concept of \emph{admissible} and \emph{maximum admissible} markings. A marking $M$ is \emph{admissible} if all pairs of marked places are in a concurrency relation. $M$ is \emph{maximum admissible} if it is not possible to add an additional token to the marking without destroying admissibility. This paper shows that each reachable marking in a sound acyclic free-choice workflow net must be maximum admissible and that each reachable sub-marking is admissible. We further show that neither computing the concurrency relation \cite{KovalyovEsparza,DBLP:conf/apn/PrinzKB24}, nor checking (maximum) admissibility, nor providing diagnostics require knowledge of concrete occurrence sequences. If a marking is \emph{not} admissible, it is not reachable because of places being in \emph{conflict} (i.\,e., they cannot have a token at the same time). The algorithm for deciding admissibility provides conflicting places as possibility for diagnostics. In addition, it provides places with missing tokens being necessary for the marking to be \emph{maximum} admissible.

(Maximum) admissibility is a necessary condition of reachability for a given marking, but unfortunately it is not sufficient. There are nets with markings that are (maximum) admissible but \emph{not} reachable. Nevertheless, admissibility provides a good heuristic. Fortunately, concurrency is always introduced by \emph{diverging transitions}, i.\,e., transitions with multiple outputs \cite{KovalyovEsparza}. Identifying such diverging transitions for a marking implies the existence of an occurrence net being a subgraph (``\emph{run}'') of the net in the case of sound acyclic free-choice nets. Such a run only diverges in transitions and contains all marked places of the (maximum) admissible marking. The overall sufficient approach eventually checks a given marking for admissibility and computes all diverging \emph{points} (i.\,e., diverging transitions \emph{and} diverging \emph{places}). In doing this, it uses information of the \emph{post-dominance frontier} \cite{DBLP:journals/toplas/CytronFRWZ91} of each node (a concept in compiler construction). Then, the algorithm checks if a diverging transition exists that leads to the marking. As a consequence, the algorithm does not rely on the examination of a concrete occurrence sequence, and, therefore, it has a quadratic computational complexity in the worst case. Another benefit of the approach is that it can explain why a marking is (not) reachable rather than just deciding reachability. This is achieved by detailed diagnostics provided by the approach.

The class of Petri nets being examined in this paper regarding reachability is, of course, limited. Nevertheless, two main reasons explain why investigating this class of nets is important: %
\begin{enumerate}
	\item Industrial business process models strongly correlate with free-choice workflow nets \cite{DBLP:journals/is/FavreFV15}. For such process models, soundness is an important minimum quality criterion \cite{DBLP:conf/edoc/DongenMA06} and can be checked in cubical computational time complexity with detailed diagnostic information \cite{DBLP:journals/is/PrinzCH25}. Although requiring \emph{acyclic} nets limits the applicability of the approach in practice, most industrial process models in prominent datasets are acyclic \cite{DBLP:journals/csimq/PrinzA21}. In summary, sound acyclic free-choice workflow nets are an interesting class of nets for an industrial setting.
	\item There is a trend investigating \emph{free-choice nets with a home cluster}. Van der Aalst \cite{DBLP:journals/fac/AalstHHSVVW11} showed that such nets with a home cluster strongly correlate with \emph{perpetual nets}. When cutting a perpetual net on its initial marking, it is unrolled to a sound free-choice workflow net. We are confident that the new method of \emph{loop decomposition} for industrial process models \cite{DBLP:conf/apn/PrinzKB24,DBLP:journals/is/PrinzCH25,DBLP:conf/bpm/PrinzCH22} can be mapped to sound free-choice workflow nets. This would further allow to separate a sound free-choice workflow net into a set of acyclic sound free-choice workflow nets while retaining the original net's behavior. Eventually, this will fill the gap to extend the here presented approach to the class of free-choice nets with a home cluster.
\end{enumerate}%
Resulting from these two main reasons, this paper is an important step to achieve a low computational polynomial time complexity to decide (sub-marking) reachability while providing diagnostics at the same moment. Furthermore, although the approach is introduced on acyclic \emph{simple} free-choice workflow nets, we briefly show that deciding reachability for acyclic \emph{extended} free-choice worklow nets does not make much difference.

The remainder of this paper is structured as follows: \Cref{sec:Preliminaries} introduces basic concepts of Petri nets, markings, reachability, paths, and soundness. Maximum admissible markings and their application are discussed in \cref{sec:MaximalMarkings} with a discussion about the output of the presented algorithm regarding diagnostics. Admissibility is then used in \cref{sec:AcyclicReachability} to finally decide reachability in sound acyclic free-choice workflow nets by introducing diverging points with a strong focus on their algorithmic derivation and diagnostics. Finally, \cref{sec:Conclusion} concludes this paper.

\section{Preliminaries}
\label{sec:Preliminaries}

This paper uses standard Petri net notions, which are provided in the following. We also recall the \emph{Path-to-End Theorem} for simple free-choice nets.

\subsection{Multisets, Petri Nets, and Paths}

$\Bag{A}$ is the set of all \emph{multisets} over some set $A$. For a multiset $b \in \Bag{A}$, $b(a)$ denotes the number of times element $a \in A$ appears in $b$. For example, $b_1 = \multiset{} = \emptyset$, $b_2 = \multiset{x,x,y}$, $b_3 = \multiset{x,y,z}$, $b_4 = \multiset{x,x,y,x,y,z}$, and $b_5 = \multiset{x^3,y^2,z}$ are multisets over the set $A = \Set{x,y,z}$. $b_1$ is the empty multiset, $b_2$ and $b_3$ consist of three elements, and $b_4 = b_5$, i.\,e., the ordering of elements is irrelevant and $b_5$ uses a more compact notation for repeating elements. The standard set operators can be extended to multisets, e.\,g., $x \in b_2$, $b_2 \mcup b_3 = b_4$, $b_5 \setminus b_2 = b_3$, etc.

\begin{definition}[Petri Nets]
\label{def:PetriNet}
A \emph{Petri net} (or simply a \emph{net}) $N$ is a triple $(P,T,F)$ with $P$ and $T$ being finite, disjoint sets of \emph{places} and \emph{transitions}, and $F \subseteq (P \times T) \cup (T \times P)$ is the \emph{flow} relation. \defend
\end{definition} 

$P \cup T$ can be interpreted as \emph{nodes} and $F$ as \emph{edges} between those nodes. %
For $x \in P \cup T$, $\preset{x} \coloneqq \Set{p \given (p,x) \in F}$ is the \emph{preset} of $x$ (all directly preceding nodes) and $\postset{x} \coloneqq \Set{s \given (x,s) \in F}$ is the \emph{postset} of $x$ (all directly succeeding nodes). Each node in $\preset{x}$ is an \emph{input} of $x$ and each node in $\postset{x}$ is an \emph{output} of $x$. The preset and postset of a set of nodes $X \subseteq P \cup T$ is defined as $\preset{X} \coloneqq \bigcup_{x \in X} \preset{x}$ and $\postset{X} \coloneqq \bigcup_{x \in X} \postset{x}$, respectively. $N$ is \emph{proper} iff $\forall t \in T\colon \preset{t} \not = \emptyset \land \postset{t} \not = \emptyset$. $N$ is \emph{(extended) free-choice} iff $\forall t_1,t_2 \in T\colon \preset{t_1} \cap \preset{t_2} \not = \emptyset \implies \preset{t_1} = \preset{t_2}$. $N$ is \emph{simple} free-choice iff $\forall p \in P\colon |\postset{p}| \geq 2 \implies \preset{(\postset{p})} = \Set{p}$, i.\,e.,  $\forall t_1,t_2 \in T\colon \preset{t_1} \cap \preset{t_2} \not = \emptyset \implies \{ p \} = \preset{t_1} = \preset{t_2}$ \cite{FreeChoicePetriNets}. 

\emph{Without loss of generality, this paper focuses on \emph{simple} free-choice nets as Murata \cite{Best1987,DBLP:journals/pieee/Murata89} presents a linear time transformation algorithm of extended to simple and simple to extended free-choice nets.} \cref{subsec:ExtendedFC} discusses how extended free-choice nets can be investigated with these transformations.

A \emph{path} $(n_1,\ldots,n_m)$ is a sequence of nodes $n_1,\ldots,n_m \in P \cup T$ with $m \geq 1$ and $\forall i \in \Set{ 1, \ldots, {m-1} }\colon \; n_i \in \preset{n_{i+1}}$. Note that places and transitions alternate on paths. $\pathset{(n_1,\ldots,n_m)}$ depicts the set of all nodes on the path. If all nodes of a path are pairwise different, the path is \emph{acyclic}; otherwise, it is \emph{cyclic}. $\Paths{x}{y}$ denotes the set of all paths between nodes $x$ and $y$, where $x,y \in P \cup T$. $N$ is \emph{acyclic} if all its paths are acyclic. In the nets shown here, circles represent places, rectangles transitions, and directed arcs represent flows as done in \cref{fig:example}.

\begin{definition}[Workflow Nets, FC-WF-Nets, and AFC-WF-Nets]
\label{def:WorkflowNet}
A \emph{workflow net} $N = (P,T,F,i,o)$ is a net $(P,T,F)$ with $i,o \in P$, $\preset{i} = \postset{o} = \emptyset$. $i$ is the \emph{source} and $o$ is the \emph{sink} of $N$. All nodes are on a path from $i$ to $o$. If $N$ is (simple) free-choice, then $N$ is called a \emph{\FW{}}. If $N$ is acyclic (simple) free-choice, $N$ is called an \emph{\AFW{}}. \defend
\end{definition}

\noindent This paper focuses on \AFWs{}.

\subsection{Markings, Reachability, Properties, and Soundness}

The behavior of nets is defined via \emph{markings}, which describe the number of \emph{tokens} on places in a specific state.

\begin{definition}[Marking]
\label{def:Marking}
A \emph{marking} $M$ of a net $N = (P,T,F)$ is a multiset of places, $M \in \Bag{P}$. $(N,M)$ is a \emph{marked net}. $\markingset{M} \coloneqq \Set{x \in M}$ depicts the set of marked places of $M$.  \defend
\end{definition}

Transitions whose input places all have tokens are \emph{enabled} in a marking and can be fired, leading to the net's semantics:

\begin{definition}[Enabledness, Firing, and Reachability]
\label{def:Semantics}
Let $(N,M)$ be a marked net $N=(P,T,F)$. A transition $t \in T$ is \emph{enabled} in $M$, denoted as $\enabledTrans{N}{M}{t}$, iff every place $\preset{t}$ contains at least one token in $M$, $\preset{t} \subseteq M$. $\allEnabled{N}{M} \coloneqq \Set{t \in T\given \enabledTrans{N}{M}{t}}$ is the set of enabled transitions in $(N,M)$.

If $t$ is enabled in $M$, then $t$ may \emph{fire}, which removes one token from each of $t$'s input places and adds one token to each of $t$'s output places. $M' = (M \setminus \preset{t}) \mcup \postset{t}$ is the marking resulting from firing $t$ in $(N,M)$. $\step{N}{M}{t}{M'}$ denotes that $t$ is enabled in $(N,M)$ and firing $t$ would result in $(N,M')$.

A sequence $\sigma = \langle t_1,t_2,\ldots,t_n \rangle \in T^*$ is an \emph{occurrence sequence} if there are markings $M_1,$ $M_2,\ldots,M_{n+1}$ such that $\forall 1 \leq i \leq n\colon \step{N}{M_i}{t_i}{M_{i+1}}$. For applying such a sequence $\sigma$, we write $\step{N}{M_1}{\sigma}{M_{n+1}}$. We say that $t_1,t_2,\ldots,t_n$ \emph{occur} in $\sigma$. 

A marking $M'$ is \emph{reachable} from the marking $M$ if there is an occurrence sequence $\sigma$ such that $\step{N}{M}{\sigma}{M'}$. $\reachableM{N}{M} \coloneqq \Set{M' \in \Bag{P} \given \exists \sigma \in T^*\colon \step{N}{M}{\sigma}{M'}}$ denotes the set of all reachable markings of $(N,M)$. 

If $N=(P,T,F,i,o)$ is a workflow net, $\multiset{i}$ is its \emph{initial marking} and $\multiset{o}$ is its \emph{final marking}. \defend
\end{definition}

\noindent Next, we list some behavioral properties important for analysis:

\begin{definition}[Live, Bounded, Safe, and Dead]
\label{def:LiveSafe}
Let $(N,M)$ be a marked net $N=(P,T,F)$. %
$(N,M)$ is \emph{live} iff for every reachable marking $M' \in \reachableM{N}{M}$ and for every transition $t \in T$, there is a reachable marking $M'' \in \reachableM{N}{M'}$, which enables $t$. %
$(N,M)$ is $k$-\emph{bounded} iff for every reachable marking $M' \in \reachableM{N}{M}$ and for every place $p \in P\colon M'(p) \leq k$. %
$(N,M)$ is \emph{bounded} iff there is a $k$ such that $(N,M)$ is $k$-bounded. %
$(N,M)$ is \emph{safe} iff $(N,M)$ is 1-bounded.

A place $p \in P$ is \emph{dead} in $(N,M)$ iff $\forall M' \in \reachableM{N}{M}\colon M'(p) = 0$. %
A transition $t \in T$ is \emph{dead} in $(N,M)$ iff $\forall M' \in \reachableM{N}{M}\colon t \notin \allEnabled{N}{M'}$. \defend
\end{definition}

\noindent Workflow nets can be \emph{sound} \cite{DBLP:conf/apn/Aalst97}:

\begin{definition}[Soundness]
\label{def:Soundness}
A workflow net $N=(P,T,F,i,o)$ with its initial marking $\multiset{i}$ and its final marking $\multiset{o}$ is \emph{sound} iff %
\begin{enumerate}[label=(\arabic*)]
	\item $\forall M \in \reachableM{N}{\multiset{i}}\colon \; \multiset{o} \in \reachableM{N}{M}$,
	\item $\forall M \in \reachableM{N}{\multiset{i}}\colon \; (o \in M \implies M = \multiset{o})$, and
	\item there is no \emph{dead} transition in $N$: $\forall t \in T \; \exists M,M' \in  \reachableM{N}{\multiset{i}}\colon \step{N}{M}{t}{M'}$.
\end{enumerate}

\noindent An equivalent definition of soundness is that the marked short-circuited net $(N',\multiset{i})$ of $N$, where $N'$ is defined as
$
	(P, T \cup \Set{ t }, F \cup \Set{ (o,t), (t,i) })
$ 
for a transition $t\notin T$, is live and bounded \cite{DBLP:conf/apn/Aalst97}. \defend
\end{definition}

\noindent Sound free-choice workflow nets are further \emph{safe}:

\begin{lemmait}[Safeness]
\label{lemma:Safeness}
Sound free-choice workflow nets are safe. \defend
\end{lemmait}

\begin{proof}
See Van der Aalst \cite{DBLP:conf/bpm/Aalst00} or Verbeek et al.\@ \cite{DBLP:journals/cj/VerbeekBA01}. \proofend
\end{proof}

Sound simple free-choice workflow nets ensure that each path from an arbitrary place to the sink place $o$ contains at most one token \cite{DBLP:conf/apn/PrinzKB24}: 

\begin{theoremit}[Path-to-End Theorem]
\label{theorem:PathToEnd}
Let $W=(P,T,F,i,o)$ be a simple \FW{}. It holds: %
On all paths $\rho$ from any place $p$ to $o$ in any reachable marking $M \in \reachableM{W}{\multiset{i}}$, the sum of all tokens on all places of $\rho$ is at most 1, i.\,e., %
\begin{equation*}
	\forall p \in P \; \; \forall \rho \in \Paths{p}{o} \; \; \forall M \in \reachableM{W}{\multiset{i}}\colon |\pathset{\rho} \cap \markingset{M}| \leq 1. \tag*{\defend}
\end{equation*}
\end{theoremit}

\begin{proof}
\proofsize See Theorem~2 in Prinz et al.\@ \cite{DBLP:conf/apn/PrinzKB24}. Note that this is a special case of Lemma~5.11 in Van der Aalst \cite{DBLP:journals/fuin/Aalst21}.\proofend
\end{proof}

\section{(Maximum) Admissible Markings}
\label{sec:MaximalMarkings}

In the following, we describe, argue, and prove why reachability and \emph{concurrency} strongly interact in sound \AFWs{}. This interaction results in the new concepts of \emph{admissibility} and \emph{maximum admissibility}. Furthermore, we discuss how (non-)admissibility can be used to provide diagnostics on why a marking is (not) reachable.

\subsection{Concurrency and Admissibility}

We will show that a given marking can only be reachable if all its marked places are pairwise concurrent. Two places are concurrent if there is a reachable marking with tokens on both places:

\begin{definition}[Concurrency]
\label{def:Concurrency}
Let $N = (P,T,F,i,o)$ be a sound \AFW{}. Two places $x,y \in P$ are \emph{concurrent} in $N$, denoted $\concurrent{x}{y}$, iff $\exists M_{\parallel} \in \reachableM{N}{\multiset{i}}\colon \multiset{ x,y } \subseteq M_{\parallel}$. $\concurrentSet{x} \coloneqq \Set{ y \in P\given \concurrent{x}{y} }$ denotes the set of all places to which $x$ is concurrent. Due to safeness of sound \AFWs{}, a place $x$ is not concurrent to itself. \defend
\end{definition}

\noindent For example, in \cref{fig:example}, $\concurrent{p9}{p10}$ and $\concurrent{p5}{p12}$ as well as $\concurrentSet{p15} = \Set{ p6 }$.

Sound \AFWs{} have benefits regarding their complexity during analysis. One of them is that two concurrent places must not have a path between them: %
\begin{lemmait}
\label{lemma:PathAbsence}
\normalfont
Let $N=(P,T,F,i,o)$ be a sound \AFW{} with two places $x,y \in P$, $x \not = y$. %
\begin{gather*}
	\concurrent{x}{y} \quad \implies \quad \Paths{x}{y} = \emptyset \; \land \; \Paths{y}{x} = \emptyset. 
\end{gather*} %
The concurrency relation $\parallel$ is symmetric and irreflexive in sound \AFWs{}. \defend
\end{lemmait}

\begin{proof}
\proofsize
See Prinz et al.\@ \cite{DBLP:conf/apn/PrinzKB24} (Cor.~4). %
The interested reader can also use the Path-to-End \cref{theorem:PathToEnd} to confirm the path absence of concurrent nodes. 
\proofend
\end{proof}

\noindent For example, $p9$ in \cref{fig:example} has no path to $p10$ and vice versa.

By \cref{lemma:Safeness}, sound \AFWs{} are safe. For this reason, all concurrent places $x$ and $y$, $\concurrent{x}{y}$, must be joined by a transition on all pairs of paths from $x$ and $y$ to $o$: %
\begin{lemmait}[Two Concurrent Places are Joined by Transitions]
\label{lemma:ConcurrentJointTransition}
Let $N=(P,T,F,i,o)$ be a sound \AFW{}. Furthermore, let $x,y \in P$. It holds: %
\begin{gather*}
	\concurrent{x}{y} \\
	\implies \\
	\forall \rho_x = (x_1,\ldots,x_n,o) \in \Paths{x}{o}, n \geq 1 \\
	\forall \rho_y = (y_1,\ldots,y_m,o) \in \Paths{y}{o}, m \geq 1 \\ 
	\exists i \in \Set{ 1, \ldots, n } \; \exists j \in \Set{ 1, \ldots, m }\colon \\
	\Set{x_1,\ldots,x_i} \cap \Set{y_1,\ldots,y_j} = \Set{ x_i } = \Set{ y_j } \subseteq T \tag*{\defend}
\end{gather*}
\end{lemmait}

\begin{proof}
\proofsize
Constructive proof. %
Let $x,y \in P$ with $\concurrent{x}{y}$. %
By \cref{def:Concurrency} of concurrency and from $\concurrent{x}{y}$, it follows that $\exists M \in \reachableM{N}{\multiset{i}}\colon \multiset{x,y} \subseteq M$. %
By \cref{lemma:Safeness}, $N$ is safe and, thus, $x \not = y$.

By \cref{def:WorkflowNet} of workflow nets, there are at least two  paths $\rho_x = (x_1,\ldots,x_n,o) \in \Paths{x}{o}$ and $\rho_y = (y_1,\ldots,y_m,o) \in \Paths{y}{o}$ from $x = x_1$ and $y = y_1$ to the sink $o$, respectively. %
It follows also: $\pathset{\rho_x} \cap \pathset{\rho_y} \not = \emptyset$. %
Therefore, $\rho_x$ and $\rho_y$ must have a \emph{first} common node $x_i = y_j$ with $x_i \in \pathset{\rho_x} \cap \pathset{\rho_y}$ and $\Set{x_1,\ldots,x_{i-1}} \cap \Set{y_1,\ldots,y_{j-1}} = \emptyset$ and $i \in \Set{ 1, \ldots, n }$, $j \in \Set{ 1, \ldots, m }$.

According to the proof of the Path-to-End Theorem~\ref{theorem:PathToEnd} in Prinz et al.\@ \cite{DBLP:conf/apn/PrinzKB24}, p.~136, Equation~(4), it holds for each transition on a path $\rho$ to the sink $o$ in a sound (simple) \AFW{} because of simple free-choiceness: %
\begin{equation*}
	\forall t \in (\pathset{\rho} \cap T)\colon \quad |\preset{t} \cap \pathset{\rho}| = 1 \quad \land \quad |\postset{t} \cap \pathset{\rho}| \geq 1
\end{equation*} %
For this reason, removing a token from $\rho_x$ or $\rho_y$ starting from $M$, respectively, can only be achieved by a place $p$ with $|\postset{p}| \geq 2$ (a decision). %
Since $N$ is simple free-choice, $\preset{(\postset{p})} = \Set{ p }$, any output transition of such $p$ can be fired in each reachable marking $M' \in \reachableM{N}{M}$, $p \in M'$ --- thus, also output transitions on $\rho_x$ ($\postset{p} \cap \pathset{\rho_x}$) and $\rho_y$ ($\postset{p} \cap \pathset{\rho_y}$). %
As a consequence, we can treat the tokens on both paths $\rho_x$ and $\rho_y$ starting from $x$ and $y$, respectively, to ``remain'' on $\rho_x$ and $\rho_y$, i.\,e., if a transition $t_x$ on $\rho_x$ fires, it puts a token back on $\rho_x$, and if a transition $t_y$ on $\rho_y$ fires, it puts a token back on $\rho_y$. %
Since $N$ is sound, this treatment of tokens to ``remain'' on $\rho_x$ and $\rho_y$ cannot lead to a dead transition. %
Furthermore, since $\concurrent{x}{y}$ and \cref{lemma:PathAbsence}, there is no path from $x$ to $y$ and from $y$ to $x$, i.\,e., no token can get from $x$ to $y$ and vice versa.

Now, there are two possibilities for the first common node $x_i = y_j$: %
\begin{description}
	\item[$x_i \in P$:] Without loss of generality, once the token of $\rho_x$ ``reaches'' $x_i$ before the token of $\rho_y$ in a reachable marking $M' \in \reachableM{N}{M}$, then in $M'$ are at least two tokens on path $\rho_y$ to $o$. %
		This contradicts the Path-to-End \cref{theorem:PathToEnd}. $\lightning$ %
		Therefore, $x_i \notin P$.
	\item[$x_i \in T$:] This is the only remaining possibility. $\checkmark$
\end{description} %
Both cases state that all pairs of paths $\rho_x$ and $\rho_y$ have a transition as a first common node. %
For this reason, this lemma holds. \proofend
\end{proof}

\cref{fig:proof-joining-transitions} shows an example of a sound \emph{extended} free-choice workflow net, for which \cref{lemma:ConcurrentJointTransition} does not hold. However, since this paper focuses on \emph{simple} free-choiceness, such cases are not possible as $t4$ and $t5$ would be merged into a single transition.

\begin{figure}[tb]
	\centering
		\includegraphics[width=0.50\textwidth]{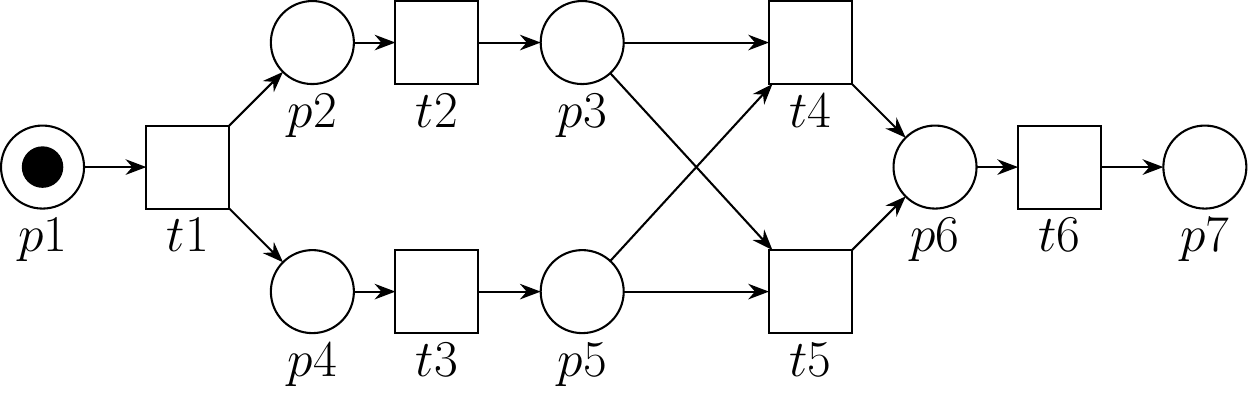}
	\caption{In sound extended free-choice workflow nets, \cref{lemma:ConcurrentJointTransition} does not hold without modifications as $p6$ is \emph{not} a joining transition.}
	\label{fig:proof-joining-transitions}
\end{figure}

The concurrency relation is symmetric by definition. The relationship is crucial for checking reachability since all marked places in a reachable marking must be in a concurrency relation by \cref{def:Concurrency}. We call a marking where all marked places are pairwise concurrent \emph{admissible}: 

\begin{definition}[Admissible Markings]
\label{def:SufficientMarkings}
Let $N$ be a sound \AFW{}. %
A marking $M_a \in \Bag{P}$ is \emph{admissible} if all different marked places in $M_a$ are pairwise concurrent: %
\begin{equation*}
	\forall x,y \in P, \, x \not = y\colon \quad \multiset{ x,y } \subseteq M_a \implies \concurrent{x}{y}. \tag*{\defend}
\end{equation*}
\end{definition} %

Each marking containing a single marked place is admissible. Following from \cref{def:SufficientMarkings}, admissibility of a marking is a necessary condition for its (sub-marking) reachability. If a marking is \emph{not} admissible, we can quickly decline its reachability.  In addition, the admissibility of markings limits the ``size'' of reachable markings by the concurrency relation. There must be \emph{maximum} admissible markings: 

\begin{definition}[Maximum Admissible Markings]
\label{def:MaximalMarking}
Let $N=(P,T,F,i,o)$ be a sound \AFW{}. %
A marking $M \in \Bag{P}$ is \emph{maximum admissible} iff 
\begin{enumerate}
  \item $\forall x,y \in M, \, x \not = y\colon \concurrent{x}{y}$ ($M$ is admissible) and 
  \item $\bigcap_{x \in M} \concurrentSet{x} = \emptyset$ (adding places of $P \setminus \markingset{M}$ to $M$ would make $M$ inadmissible). \defend
\end{enumerate}
\end{definition}

By \cref{def:SufficientMarkings}, a marking $M$ composed of concurrent places cannot be reachable for three reasons: (1) $M$ contains not enough marked places, (2) $M$ contains too many marked places, or (3) $M$ has a ``correct'' number of marked places but all places cannot be marked at the same time. Fortunately, we can show that each \emph{reachable} marking in a sound \AFW{} is \emph{maximum admissible}, i.\,e., removes reasons (1) and (2):

\begin{theoremit}[Reachable Markings are Maximum Admissible]
\label{theorem:ReachableMarkingsMaximal}
Let $N$ be a sound \AFW{}. All reachable markings $M_r$ from the initial marking $\multiset{i}$ are maximum admissible. \defend
\end{theoremit}

\begin{proof}
\proofsize
Let $N=(P,T,F,i,o)$ be a sound \AFW{}. %
Furthermore, let $M_r=\multiset{ p_1,\ldots,p_m }$, $m \geq 1$, be a reachable marking $M_r \in \reachableM{N}{\multiset{i}}$. %
The proof is done by contradiction: $M_r$ is \emph{not} maximum admissible. %
As $M_r$ is admissible by \cref{def:SufficientMarkings} but $M_r$ is not \emph{maximum} admissible, by \cref{def:MaximalMarking}, there must be a place $p \in (P \setminus \markingset{M_r})$ not in $M_r$ being concurrent to all marked places in $M_r$: %
\begin{equation}
  \exists p \in (P \setminus M_r) \; \forall p' \in M_r\colon \concurrent{p'}{p} \label{eq:p101}
\end{equation}
Let $p$ be such a place in the following.

For each $p_i \in M_r$, $1 \leq i \leq m$, there is a path $\rho_i \in \Paths{p_i}{o}$ to the sink $o$ by \cref{def:WorkflowNet} of workflow nets. %
In addition, let $\rho_p \in \Paths{p}{o}$ be a path to $o$, which contains $p$ by \cref{def:WorkflowNet}. %
By \cref{theorem:PathToEnd}, for each $\rho_i$ there is exact one token $\rho_i$ on $p_i$ in $M_r$ ($\rho_i$ is safe). %
There are exactly two cases for $\rho_p$: %
\begin{description}
	\item[Case 1:] $\exists p' \in \big( \pathset{\rho_p} \cap P \big)\colon p' \in M_r$ (there is a token on $\rho_p$ in $M_r$). %
	Thus, there is a path from $p$ to $p'$, $\Paths{p}{p'} \not = \emptyset$. %
	Since $\concurrent{p}{p'}$ by \eqref{eq:p101}, this contradicts \cref{lemma:PathAbsence} that concurrency requires the absence of paths. $\lightning$ %
	This case does not hold. 
	\item[Case 2:] $\forall p' \in \big( \pathset{\rho_p} \cap P \big)\colon p' \notin M_r$ (there is no token on $\rho_p$ in $M_r$). %
The entire situation is abstractly illustrated in \cref{fig:ProofMaximal}. %
  By \cref{lemma:ConcurrentJointTransition}, 
all such two paths $\rho_x$ and $\rho_y$, $\rho_x,\rho_y \in \Set{ \rho_1,\ldots,\rho_m,\rho_p }$, $\rho_x \not = \rho_y$, contain a transition $t$ with at least two input places $in_x \in (\pathset{\rho_x} \cap \preset{t})$ and $in_y \in (\pathset{\rho_y} \cap \preset{t})$, $in_x \not = in_y$; even for $\rho_p$ with any other path $\rho_i$, $1 \leq i \leq m$. For the moment, we say $\rho_x$ and $\rho_y$ ``converge'' in $t$. %
	By this case, $\rho_p$ is without any token in $M_r$, i.\,e., the transition(s) $t_1,\ldots,t_m$, in which $\rho_p$ converges with any of the other paths $\rho_i$, $1 \leq i \leq m$, are dead when the token(s) on those paths reach any of $t_1,\ldots,t_k$ in a marking $M' \in \reachableM{N}{M_r}$. %
	This contradicts \cref{def:Soundness} of soundness. $\lightning$ %
	This case does not hold. %
\end{description} %
\begin{figure}[tb]
	\centering 
		\includegraphics[width=0.4\textwidth]{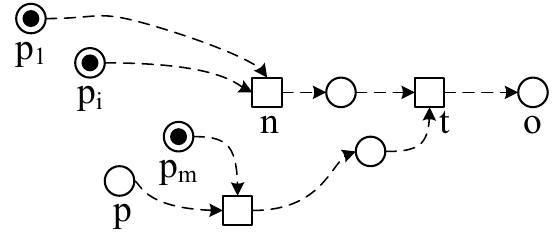}
	\caption{For each $p_i$ with a token in $M_r$, there is a path $\rho_i$ to the sink $o$.}
	\label{fig:ProofMaximal}
\end{figure}
Since both cases do not hold, the assumption contradicts. For this reason, the theorem holds that each reachable marking is maximum admissible. \proofend
\end{proof}

For example, the marking $\multiset{ p9,p10 }$ in \cref{fig:example} is \emph{not} maximum admissible since $p9$ and $p10$ have $p3$, $p17$, etc. as common concurrent places. The marking $\multiset{ p3,p5 }$ is \emph{not} maximum admissible because it is not admissible at all: $\notconcurrent{p3}{p5}$. However, the marking $\multiset{ p5,p12,p14 }$ is maximum admissible as all places, to which $p12$ is concurrent (e.\,g., $p9$, $p11$, $p18$, etc.), are not concurrent to $p5$ and $p14$.

It follows for checking if a marking is reachable in a sound \AFW{} to investigate first its maximum admissibility. By \cref{def:MaximalMarking}, checking maximum admissibility of a marking depends on two simple rules. Following these two rules, the computation of whether a given marking is (maximum) admissible or not, can be decided in $O(|P|^2 + |T|^2)$:

\begin{theoremit}[Computational Complexity of (Maximum) Admissibility]
\label{theorem:ComputationalComplexityMax}
Let $N=(P,T,F,i,o)$ be a sound \AFW{}. Deciding, whether a given marking $M_r$ is maximum admissible in $N$ or not, can be achieved in $O(|P|^2 + |T|^2)$. \defend
\end{theoremit}

\begin{proof}
\proofsize
It is possible with the algorithm of Prinz et al.\@ \cite{DBLP:conf/apn/PrinzKB24} to determine the concurrency relation by \cref{def:Concurrency} in $O(|P|^2 + |T|^2)$ for sound \AFWs{}. %
The second rule of \cref{def:MaximalMarking} states: %
\begin{equation}
	\bigcap_{x \in M_r} \concurrentSet{x} = \emptyset.
\end{equation}
Since $x \in M_r$ is concurrent to all other marked places of $M_r$ except to itself, $x$ is not in the intersection of all concurrency sets because $x$ is missing in its own concurrency set $\concurrentSet{x}$. %
For this reason, temporarily adding $x$ to its own concurrency set leads to a combination of both rules: %
\begin{equation}
	\bigcap_{x \in M_r} \big( \concurrentSet{x} \cup \Set{ x } \big) = M_r. \label{eq:p201}
\end{equation}
Once the concurrency relation is computed and the concurrency sets of all places are stored within a computationally efficient data structure for mathematical sets being able to compute the intersection in constant time, e.\,g., by a BitSet in Java, \cref{eq:p201}
can be checked in linear time, $O(|P|)$.

In summary, the overall computational complexity is dominated by the computation of the concurrency relation and can, therefore, be achieved in $O(|P|^2 + |T|^2)$. Furthermore, if the concurrency relation is stored, checking (maximum) admissibility can be achieved in linear time, $O(|P|)$. \proofend
\end{proof}

\subsection{Algorithm and Output}
\label{subsec:AdmissibilityOutput}

\cref{def:MaximalMarking} of maximum admissible markings and the proof of \cref{theorem:ComputationalComplexityMax} allows for defining a straight-forward algorithm to check (maximum) admissibility, which is stated in \cref{algo:CheckMaximumAdmissibility}. This algorithm requires a sound \AFW{} $N$ and a marking $M_r$ (to investigate) as inputs. \cref{al:ma:1} computes the concurrency relation. This is achieved with the \emph{Concurrent Paths (CP)} algorithm being stated in \cref{algo:DetermineConcurrency}. It is currently the algorithm having the best asymptotic computational time complexity for sound \AFWs{}.  Instead of investigating each concurrent pair of places, the \emph{CP} algorithm investigates sets of nodes as paths. It requires the computation of the $\mathit{HasPath}$ relation (\cref{algo:HasPath}) as a pre-computation step, whereas this relation just states the sets of nodes to which a node has paths. The reader can find more information in \cite{DBLP:conf/apn/PrinzKB24}. Setting the focus back to \cref{algo:CheckMaximumAdmissibility}, \cref{al:ma:2} assigns the set of all places $P$ of $N$ to $M_p$ as maximal possibility of concurrent places and, then, incrementally reduces $M_p$ regarding each $x \in M_r$ in Lines~\ref{al:ma:3}--\ref{al:ma:4} following the proof of \cref{theorem:ComputationalComplexityMax}.

\begin{algorithm}[t]
\caption{Checking (maximum) admissibility of a given marking $M_r \in \Bag{P}$ for a sound \AFW{} $N=(P,T,F,i,o)$.}
\label{algo:CheckMaximumAdmissibility}
\SetKwProg{Fn}{Function}{}{end}
\Fn{checkMaximumAdmissibility($N$, $M_r$)}{
  $\parallel \; \gets$ determineConcurrency($N$)\; \nllabel{al:ma:1} 
	$M_p \gets P$\;\nllabel{al:ma:2} 
	\For{$x \in M_r$}{\nllabel{al:ma:3}
		$M_p \gets M_p \cap \big( \concurrentSet{x} \cup \Set{x}\big)$\;
	}\nllabel{al:ma:4} 
	\uIf{$M_r = M_p$}{\nllabel{al:ma:ma}
		\KwRet{$(\text{maximum admissible}, \emptyset, \emptyset)$}
	}\Else{
		\uIf{$M_r \subset M_p$}{\nllabel{al:ma:a}
			\KwRet{$(\text{admissible}, M_p \setminus M_r, \emptyset)$}
		}\Else{
			\KwRet{$(\text{not admissible}, M_p \setminus M_r, M_r \setminus M_p)$}\nllabel{al:ma:na}
		}
	}	
}
\end{algorithm}

There are three kinds of outputs of \cref{algo:CheckMaximumAdmissibility}: marking $M_r$ is either (1) \emph{maximum admissible} in the case $M_r = M_p$ (\cref{al:ma:ma}), (2) \emph{admissible} in the case $M_r \subset M_p$ (\cref{al:ma:a}), or (3) \emph{not admissible} (\cref{al:ma:na}). In the first case~(1), there is no derivation between $M_r$ and $M_p$, thus, the algorithm just returns the decision. The second case~(2) is admissible but has missing places, i.\,e., the marking is not maximal in the sense of admissibility. For this reason, the algorithm returns (besides its decision) also the set of places $M_p \setminus M_r$ that is (potentially) missing. For the last case~(3), $M_r$ is not admissible so that $M_r$ contains at least one marked place that is not concurrent to another marked place within $M_r$, i.\,e., at least two marked places are in \emph{conflict} within $M_r$. Such places can be identified by removing a possible maximum admissible marking ($M_p$) from $M_r$. Furthermore, there can also be missing places for this case as for case~(2). As a consequence, the algorithm returns the decision, the set of (potentially) missing places $M_p \setminus M_r$, and the set of conflicting places $M_r \setminus M_p$.

\begin{algorithm}[t]
\caption{The \emph{Concurrent Paths (CP)} algorithm: Deriving the concurrency relation $\parallel$ for a sound \AFW{} $N=(P,T,F,i,o)$ as adjancency list (adapted from \cite{DBLP:conf/apn/PrinzKB24}).}
\label{algo:DetermineConcurrency}
\SetKwProg{Fn}{Function}{}{end}
\Fn{determineConcurrency($N$)}{
	\tcp{Initialize}
	$\parallel \gets \emptyset$\;
	\lFor{$x \in P \cup T$}{
		$\concurrentSet{x} \gets \emptyset$
	}
	$\mathit{HasPath} \gets$ computeHasPath($N$)\;
	\tcp{Compute}
	$I \gets \emptyset$\;
	\For{$t \in T$}{
		\lFor{$x \in \postset{t}$}{
			$I(x) \gets I(x) \cup \big(\postset{t} \setminus \Set{x}\big)$
		}
	}
	\For{keys $x$ of $I$}{
		\For{$y \in I(x)$}{
			$R^{\conj{y}}_{x} \gets \mathit{HasPath}(x) \setminus \mathit{HasPath}(y)$\;
			\lFor{$a \in R^{\conj{y}}_{x}$}{
				$\concurrentSet{a} \gets \concurrentSet{a} \cup \big( \mathit{HasPath}(y) \setminus \mathit{HasPath}(a) \big)$
			}
		}
	}
	\KwRet{$\parallel$}  
}
\end{algorithm}

\begin{algorithm}[t]
\caption{Computing the $\mathit{HasPath}$ relation for a sound \AFW{} $N=(P,T,F,i,o)$ as adjacency list (adapted from \cite{DBLP:conf/apn/PrinzKB24}).}
\label{algo:HasPath}
\SetKwProg{Fn}{Function}{}{end}
\Fn{computeHasPath($N$)}{
	\tcp{Initialize}
	$\mathit{HasPath} \gets \emptyset$\;
	\lFor{$x \in P \cup T$}{
		$\mathit{HasPath}(x) \gets \emptyset$
	}
	$L \gets P \cup T$ in reverse topological order starting from $o$\;
	\tcp{Compute}
	\lFor{$x \in L$}{
		$\mathit{HasPath}(x) \gets \Set{x} \cup \bigcup_{s \in \postset{x}} \mathit{HasPath}(s)$
	}
	\KwRet{$\mathit{HasPath}$}
}
\end{algorithm}

For example, applying the net of \cref{fig:example} with the marking $\multiset{p5,p12,p14}$ to \cref{algo:CheckMaximumAdmissibility} will lead to the output ``\emph{maximum admissible}'' without the necessity for giving further feedback. Applying the marking $\multiset{p9,p10}$ and the same net to the algorithm will lead to the output ``\emph{admissible}'' and the set $\Set{p2,p3,p11,p12,p13,p14,$ $p16,p17,p18}$ of (potentially) ``missing'' places. This output can be used as a diagnostic so that other algorithms or users can be informed that the marking $\multiset{p9,p10}$ is \emph{admissible} but will not be reachable without other places being marked at the same time. Those other places could be colored, e.\,g., in green. \cref{fig:examplep9p10} shows an example on how the information can be visualized. %
\begin{figure}[t]
	\centering
	\includegraphics[width=0.8\textwidth]{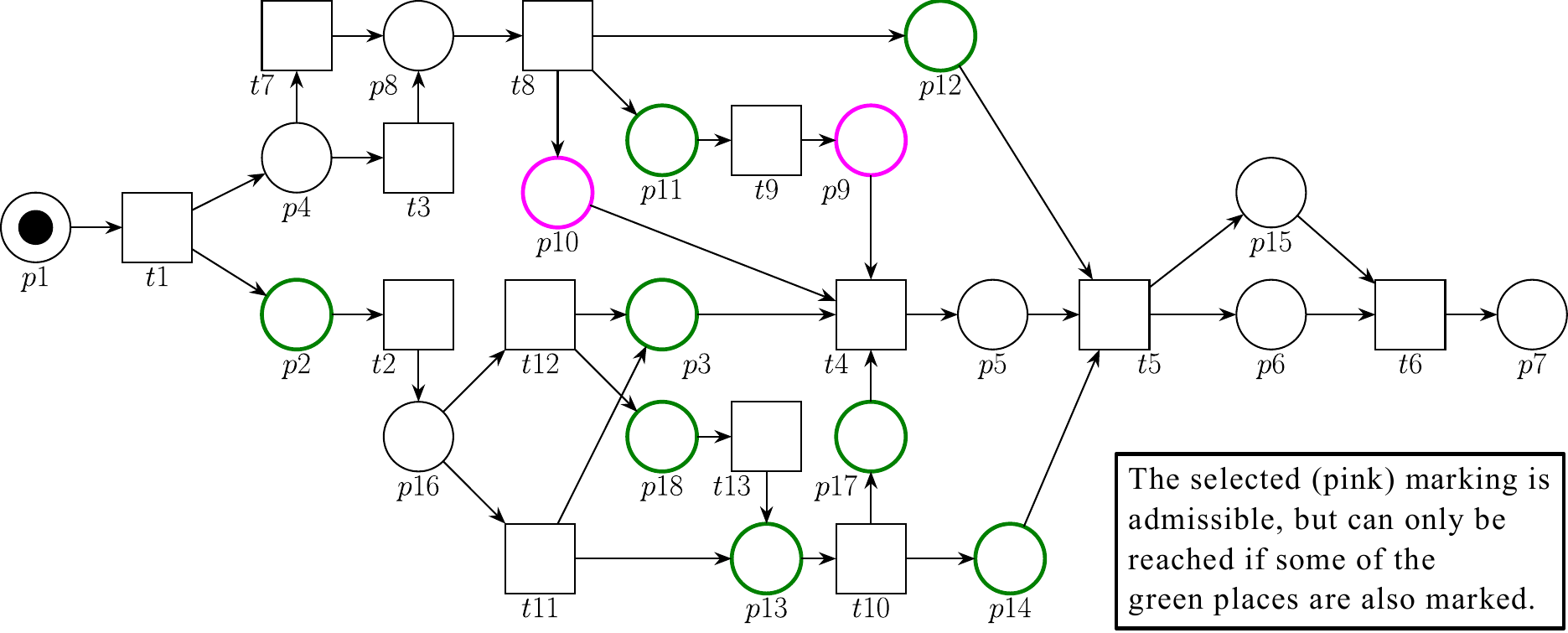}
	\caption{The net of \cref{fig:example} with selected marking $\multiset{p9,p10}$ (highlighted in pink). The marking is \emph{admissible} by \cref{algo:CheckMaximumAdmissibility}. At least one of the places highlighted in green will be marked at the same time as $\multiset{p9,p10}$ in a reachable marking as this is not \emph{maximum} admissible.}
	\label{fig:examplep9p10}
\end{figure}
\begin{figure}[t]
	\centering
		\includegraphics[width=0.8\textwidth]{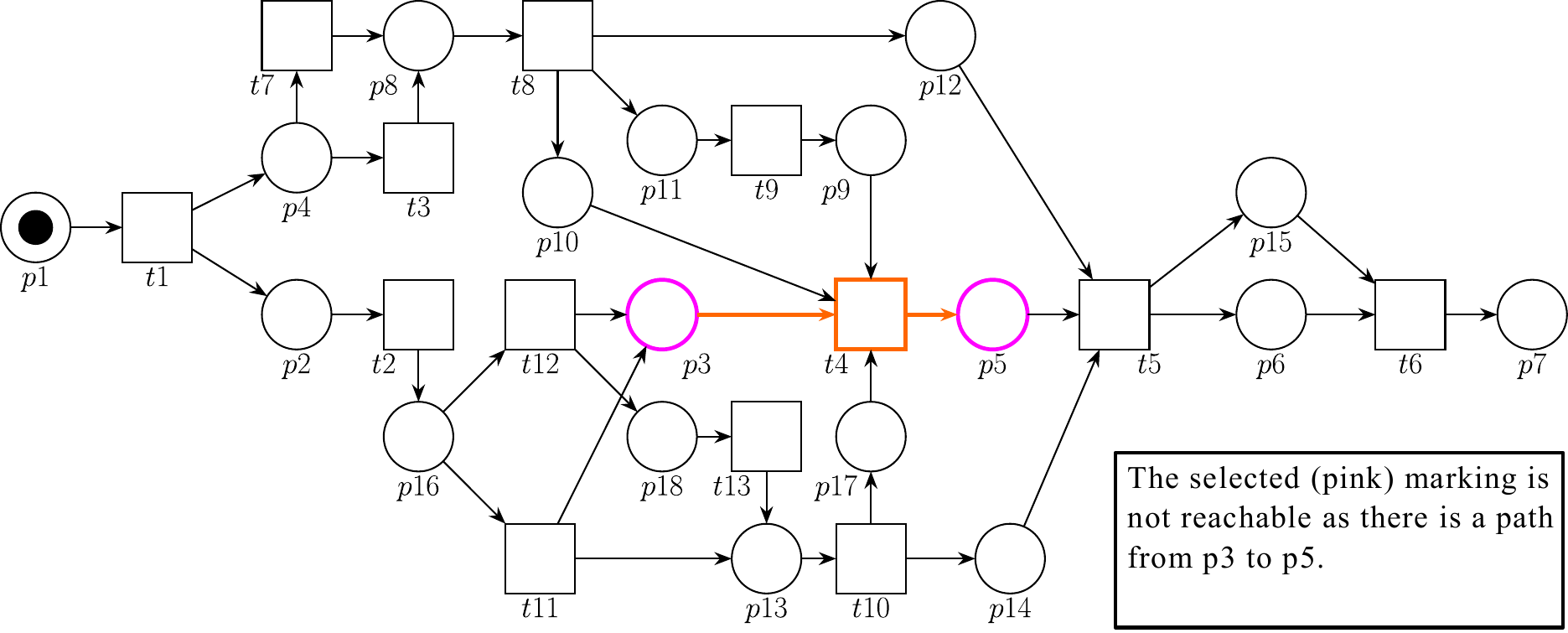}
	\caption{The net of \cref{fig:example} with selected marking $\multiset{p3,p5}$ (highlighted in pink). The marking is \emph{not admissible} by \cref{algo:CheckMaximumAdmissibility} since $p3$ is not concurrent to $p5$ as there is a path from $p3$ to $p5$ highlighted in orange.}
	\label{fig:examplep3p5}
\end{figure}
The marking $\multiset{p3,p5}$ of \cref{fig:example} applied to \cref{algo:CheckMaximumAdmissibility} will lead to a ``\emph{not admissible}'' output with $M_p$ in the algorithm gets: %
\begin{align*}
	M_p = & \; P \cap \big( \concurrentSet{p3} \cup \Set{p3} \big) \cap \big( \concurrentSet{p5} \cup \Set{p5} \big) \\
	    = & \; P \cap \Set{p3,p4,p8,p9,p10,p11,p12,p13,p14,p17,p18} \cap \Set{p5,p12,p14} \\
			= & \; \Set{p12,p14}. 
\end{align*} %
As a consequence, the algorithm returns $M_p \setminus M_r = \Set{p12,p13}$ as (potentially) missing places and $M_r \setminus M_p = \Set{p3,p5}$ as conflicting places since $\notconcurrent{p3}{p5}$. Investigating this output more intensive, there are three reasons why $p3$ cannot be concurrent to $p5$ in general: (1) there is a path from $p3$ to $p5$ (such as in this case, \cref{lemma:PathAbsence}), (2) there is a path from $p5$ to $p3$ (\cref{lemma:PathAbsence}), or (3) there are no paths between $p3$ and $p5$ and, therefore, they are mutually exclusive (compare, e.\,g., \cite{DBLP:conf/edoc/PrinzWH24}). Cases~(1) and (2) can be identified easily with a depth-first search (or the $\mathit{HasPath}$ relation of \cref{algo:HasPath}) for further diagnostics (e.\,g., to color both places and the path between them as illustrated in \cref{fig:examplep3p5} in orange). Thus, the remaining case~(3) is for free. For such a case, take the sound \AFW{} in \cref{fig:example-divpoints-2-p3p5p7} with marking $\multiset{p3,p5,p7}$. Obviously, $p5$ and $p7$ are not concurrent but mutually exclusive (as we know from the above thoughts). However, it would be beneficial for users to know about the ``decision point'' $p4$ with the exclusive paths to $p5$ and $p7$ ``hindering'' the concurrency of both places. Fortunately, the approach in the next section will be able to identify such information.

\begin{figure}[t]
	\begin{minipage}[t]{0.58\textwidth}
		\raggedleft
			\includegraphics[width=0.95\textwidth]{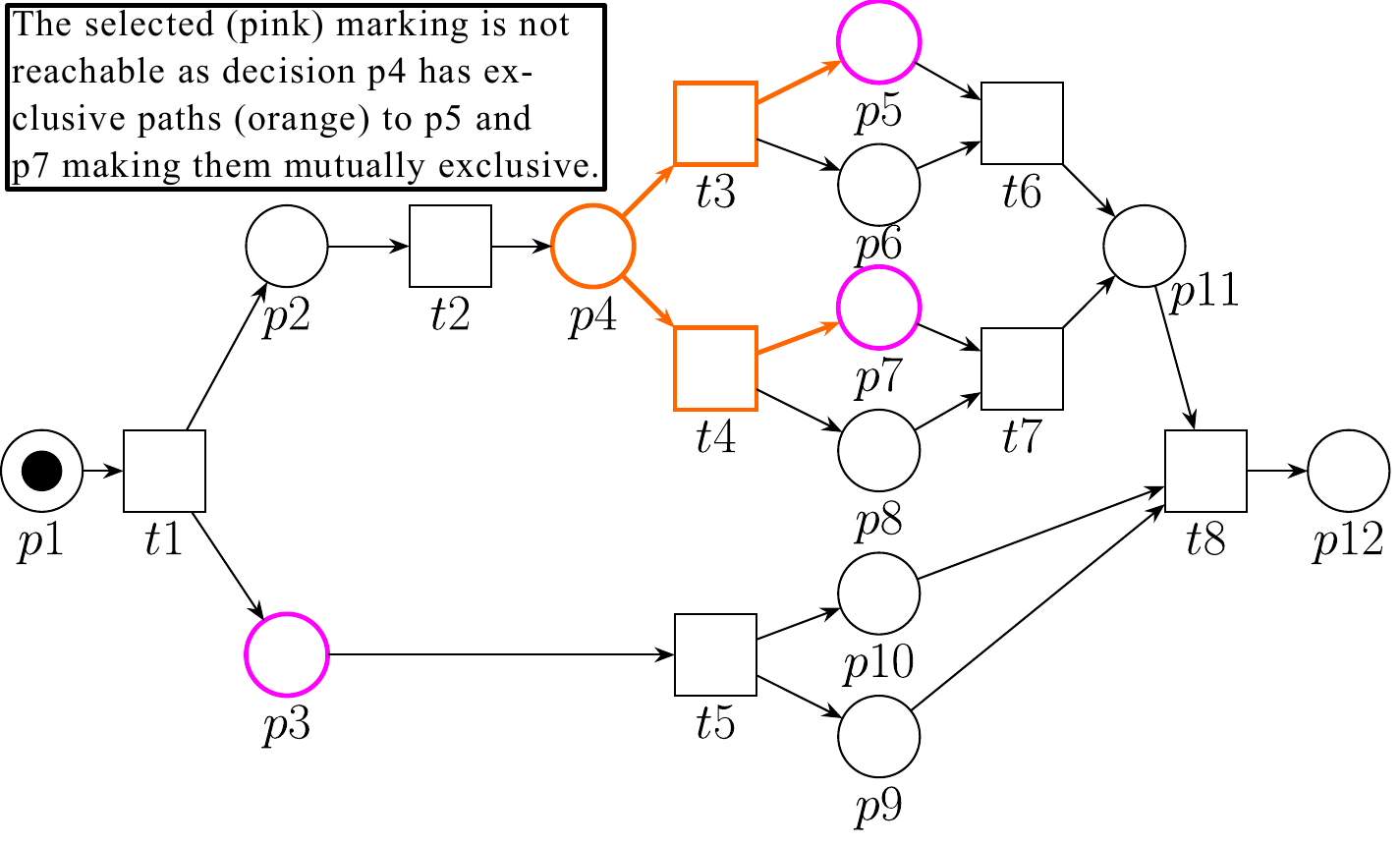}
		\caption{A sound \AFW{} with selected marking $\multiset{p3,p5,p7}$ (highlighted in pink). The marking is \emph{not admissible} by \cref{algo:CheckMaximumAdmissibility} since $p5$ is not concurrent to $p7$. Since both places do not have any path in-between them, they are mutually exclusive. For diagnostics, it would be of benefit to present the decision point $p4$ in the net making both places exclusive.}
		\label{fig:example-divpoints-2-p3p5p7}
	\end{minipage}\hfill%
	\begin{minipage}[t]{0.4\textwidth}
		\raggedright
		\includegraphics[width=0.95\textwidth]{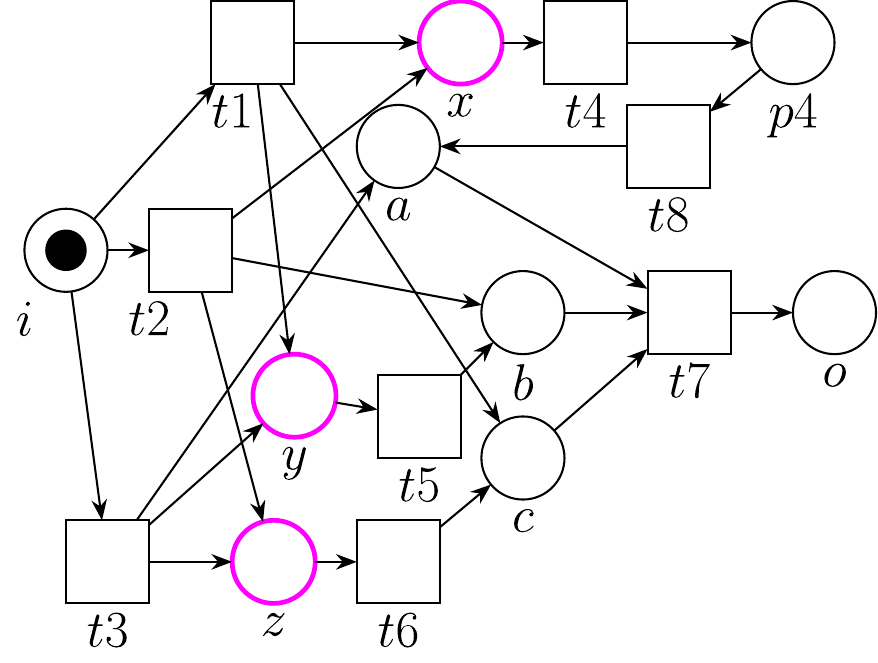}
	\caption{A net with a maximal admissible marking $\multiset{x,y,z}$ that is not reachable.}
	\label{fig:MaxNotReach}
	\end{minipage}	
\end{figure}

\section{Reachability}
\label{sec:AcyclicReachability}

Unfortunately, maximum admissibility is only a necessary condition for reachability but \emph{not} sufficient since not each maximum admissible marking is also reachable as the net in \cref{fig:MaxNotReach} demonstrates. This net contains three places $x$, $y$, and $z$ (colored in pink in the figure). They are pairwise concurrent, i.\,e., $\concurrent{x}{y}$, $\concurrent{x}{z}$, and $\concurrent{y}{z}$. The resulting marking $\multiset{x,y,z}$ is maximum admissible. However, as the reader can confirm, this marking is not reachable since the places are not concurrent to each other at the same time. As a consequence, not every maximum admissible marking is reachable. We have to confirm maximum admissible markings which are reachable as well.

In the following, the new concept of \emph{diverging points} is introduced that finally leads to an algorithm for deciding reachability while providing diagnostics at the same time. Eventually, it will be discussed briefly how extended free-choiceness can be handled with the approach.

\subsection{Run Nets and Diverging Points}

In the following, we will show the sufficient condition of reachability in sound \AFWs{}. Besides the absence of paths between concurrent places, another benefit of sound \AFWs{} is that their \emph{occurrence nets} \cite{DBLP:conf/apn/PolyvyanyyWCRH14} (i.\,e., a net, which represents how a corresponding net was executed) are simplified to \emph{run nets} \cite{DBLP:conf/edoc/PrinzWH24}. In contrast to occurrence nets, run nets are just subgraphs (sub \emph{nets}) of sound \AFWs{} and, therefore, each node can only occur once. A run net is defined as follows \cite{DBLP:conf/edoc/PrinzWH24}: 

\begin{definition}[Run Nets]
\label{def:RunNet}
Let $N = (P,T,F,i,o)$ be a sound \AFW{}. An occurrence sequence $\sigma \in T^*$ is a \emph{run} iff $\sigma$ leads from the initial marking $\multiset{i}$ to the final marking $\multiset{o}$. A net $\pi = (P_\sigma,T_\sigma,F_\sigma)$ is a \emph{run net} of a run $\sigma$ of $N$ iff %
\begin{align*}
	P_\sigma  = & \; \Big \{ \, p \in P \; \big| \; p \in \bigcup_{t \in \sigma} (\preset{t} \cup \postset{t}) \, \Big \} \\
	T_\sigma  = & \; \Set{ t \in T \mid t \in \sigma } \\
	F_\sigma  = & \; F \cap \big( (P_\sigma \cup T_\sigma) \times (P_\sigma \cup T_\sigma) \big).
\end{align*} %
$n \in \sigma$ \emph{occurs} in $\pi$, depicted as $n \in \pi$. $\RunNets{N}{M_0}$ depicts the set of all run nets of $N$. \defend
\end{definition}

\begin{figure}[t]
	\centering
		\includegraphics[width=0.80\textwidth]{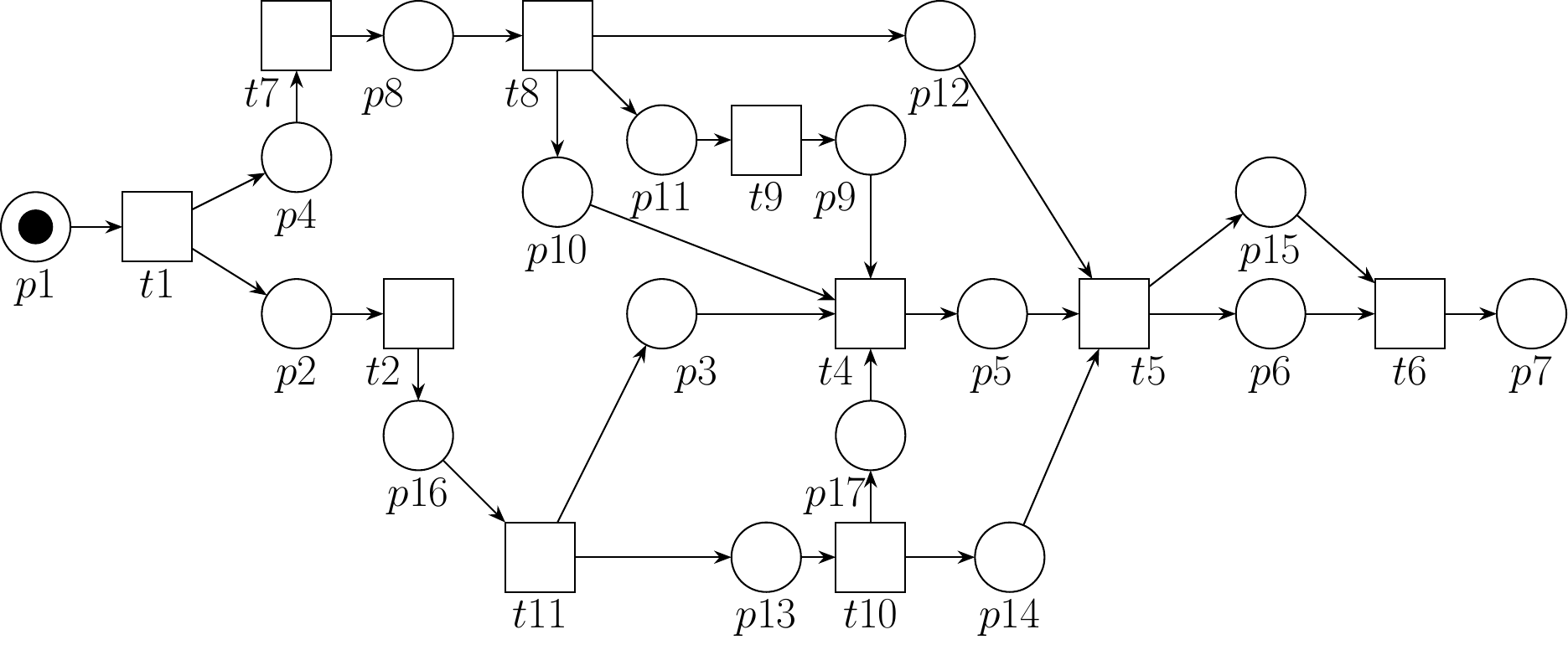}
		\vspace{-0.25cm}
	\caption{Example of a run net of the example workflow net in \cref{fig:example}.}
	\label{fig:example-run-net}
\end{figure}

\noindent If a run net $\pi = (P_\sigma,T_\sigma,F_\sigma)$ of a sound \AFW{} $N = (P,T,F,i,o)$ contains a transition $t \in T_\sigma$, then it must also contain its input and output places, $\forall t \in T_\sigma\colon (\preset{t} \cup \postset{t}) \subseteq P_\sigma$. If a run net contains a place $p \in P_\sigma$, $p \notin \Set{ i, o }$, it must also contain exactly one of its input and one of its output transitions, $\forall p \in (P_\sigma \setminus \Set{ i,o })\colon |\postset{p} \cap T_\sigma| = 1$ and $|\preset{p} \cap T_\sigma| = 1$. \Cref{fig:example-run-net} shows an example of a run net of our net in \cref{fig:example}. 

Let $M_m$ be a marking for which we want to decide reachability in a given sound \AFW{} $N$. $M_m$ must be maximum admissible by \cref{theorem:ReachableMarkingsMaximal} and, thus, all marked places of $M_m$ are pairwise concurrent. Regarding the run nets of $N$ and \cref{lemma:PathAbsence} about the absence of paths between concurrent places, there cannot be paths between any possibly concurrent places in any run net of $N$. This fact was also confirmed in Theorem~2 in Prinz et al.\@ \cite{DBLP:conf/edoc/PrinzWH24}. For this reason, deciding whether a maximum admissible or admissible sub-marking $M_m$ is reachable in $N$ corresponds with the existence of a run net that contains all marked places in $M_m$:

\begin{corollaryit}
\label{cor:ReachabilityWithRunNets}
Let $N = (P,T,F,i,o)$ be a sound \AFW{} and a marking $M_m$. It holds: %
\begin{align}
	M_m \in \reachableM{N}{\multiset{i}} \quad \iff \quad & M_m \text{ is maximum admissible } \quad \land \label{eq:c11}\\ & \exists (P_\sigma,T_\sigma,F_\sigma) \in \RunNets{N}{\multiset{i}}\colon \markingset{M_m} \subseteq P_\sigma \nonumber \\ %
	\exists M_s \in \reachableM{N}{\multiset{i}}\colon M_m \subseteq M_s \quad \iff \quad & M_m \text{ is admissible} \quad \land \label{eq:c12} \\ & \exists (P_\sigma,T_\sigma,F_\sigma) \in \RunNets{N}{\multiset{i}}\colon \markingset{M_m} \subseteq P_\sigma \tag*{\defend}
\end{align}
\end{corollaryit}

\begin{proof}
\proofsize
Constructive proof by focusing on \cref{eq:c11}: %
\begin{description}
	\item[$\implies$] $M_m \in \reachableM{N}{\multiset{i}}$. %
  Since $M_m$ is reachable from $\multiset{i}$, there must be a run net by \cref{def:RunNet} of run nets, which contains all marked places of $M_m$. %
	Furthermore, by \cref{theorem:ReachableMarkingsMaximal}, each reachable marking is maximum admissible, and, therefore, $M_m$. $\checkmark$
	\item[$\impliedby$] $M_m$ is maximum admissible $\land$ $\exists (P_\sigma,T_\sigma,F_\sigma) \in \RunNets{N}{\multiset{i}}\colon \markingset{M_m} \subseteq P_\sigma$. %
	Since $M_m$ is maximum admissible, all marked places of $M_m$ are pairwise concurrent. %
	Furthermore, there is a run net $\pi$, which contains all marked places of $M_m$. %
	For this reason, and by \cref{lemma:PathAbsence}, all marked places of $M_m$ do not have paths to each other in $N$ as well as in $\pi$. %
	Thus, all places in $M_m$ can have tokens in the same reachable marking from $\multiset{i}$ as they can appear in any order from each other. %
	As a consequence, $M_m$ is a reachable marking from $\multiset{i}$. $\checkmark$
\end{description} %
The proof of \cref{eq:c12} can be done similarly. \proofend
\end{proof}

Since we are able to decide the (maximum) admissibility  of a marking in $O(|P|^2 + |T|^2)$, we need an approach to check if there is at least one run net, which contains all marked places of a given marking. Deriving such a run net is not trivial, as decisions (places with multiple outgoing flows) can lead to an exponential growth of possible run nets. However, we do not have to derive a concrete run net --- we only have to show that such a run net must exist, or not. For this reason, we will focus on transitions in the following, which diverge the control-flow and therefore ``produce'' concurrency.

If a sound \AFW{} has a reachable marking $M_r$ with at least two marked places $x$ and $y$, then there must be at least one transition $t$ with multiple output places $\Set{x_t,y_t,\ldots}$ (as the net has just one source) and $x_t$ and $y_t$ have disjoint paths to $x$ and $y$ so that the tokens are not joint before $x$ and $y$. We call such $t$ a \emph{diverging point} of $\Set{x,y}$ as well as a \emph{diverging point} of the set $\markingset{M_r} = D$. \cref{fig:div-points} illustrates the concept of diverging points abstractly for three places $x$, $y$, and $z$. $\delta_1$ is a diverging point of $\Set{x,y,z}$ although it only has disjoint paths to $x$ and $y$. $\delta_2$ is not a diverging point of $\Set{x,y,z}$ since it only has paths to $z$. The transition $\delta_3$ is a diverging point of $\Set{x,y,z}$ as it has disjoint paths to $x$ and $z$ as well as $y$ and $z$. The source $i$ is \emph{not} a diverging point of $\Set{x,y,z}$ since each path from $i$ to them goes through $\delta_3$. The following definition describes that formally:

\begin{figure}[tb]
	\centering
		\includegraphics[width=0.60\textwidth]{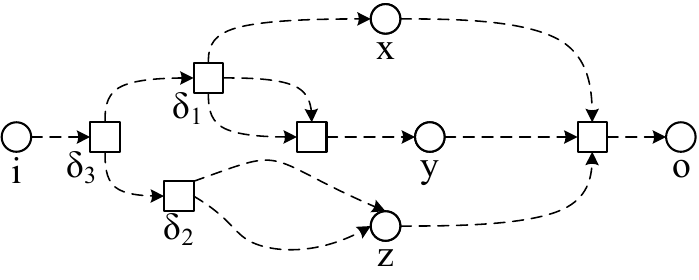}
	\caption{For the marking $\multiset{x,y,z}$ and its set $D = \markingset{\multiset{x,y,z}} = \Set{ x,y,z }$ there are diverging points $\delta_1$ (for $x$ and $y$) and $\delta_3$ (for $x$ and $z$ as well as $y$ and $z$). The source $i$ is \emph{no} diverging point of $\Set{x,y,z}$ as all paths to them pass the same output transition of $i$. $\delta_2$ is no diverging point of $\Set{x,y,z}$ as it has only paths to $z$.}
	\label{fig:div-points}
\end{figure}

\begin{definition}[Diverging Points]
\label{def:DivergingPoint}
Let $N=(P,T,F,i,o)$ be an \AFW{}. %
A node $\delta \in P \cup T$ is a \emph{diverging point} of a set $D \subseteq P \cup T$, $|D| \geq 2$, of nodes, depicted $\idivpoint{\delta}{D}$, if for at least two nodes $d_1,d_2 \in D$, $d_1 \not = d_2$, and for at least two output nodes $o_1,o_2 \in \postset{\delta}$, there are two disjoint paths $\rho_1$ from $o_1$ to $d_1$ and $\rho_2$ from $o_2$ to $d_2$. %
\[ \idivpoint{\delta}{D} \iff \exists d_1,d_2 \in D \; \exists o_1,o_2 \in \postset{\delta} \; \exists \rho_1 \in \Paths{o_1}{d_1} \; \exists \rho_2 \in \Paths{o_2}{d_2}\colon 	\pathset{\rho_1} \cap \pathset{\rho_2} = \emptyset \]
Let $\idivpoints{D} \coloneqq \Set{ \delta \in P \cup T\given \idivpoint{\delta}{D} }$ be the set of all diverging points of $D$. %
Note that we use $\idivpoints{M}$ as a short version of $\idivpoints{\markingset{M}}$ for reasons of readability for the case that $M$ is a marking.

For an output $o \in \postset{\delta}$ and the set $D$, the set $\divinfo{o}{D} \coloneqq \Set{ d \in D \mid \Paths{o}{d} \not = \emptyset }$ defines those nodes of $d \in D$ for which there is a path from $o$ to $d$ (i.\,e., $o$ \emph{reaches} $d$). %
Furthermore, $\divinfoh{\delta}{\markingset{M_m}} \coloneqq \bigcup_{o \in \postset{\delta}} \divinfo{o}{\markingset{M_m}}$ depicts the union of reachable nodes of $\delta$'s outputs. %
\defend
\end{definition}

In the \AFW{} of \cref{fig:MaxNotReach}, the set $D = \Set{x,y,z} = \markingset{[x,y,z]}$ of places has the diverging points $\idivpoints{\Set{x,y,z}} = \Set{t1,t2,t3,i}$. For example, there are disjoint paths from $t1$ to $x$ and from $t1$ to $y$. For the outputs $\Set{x,y,c}$ of $t1$, the subsets of nodes of $\Set{x,y,z}$ (they have paths to) are $\divinfo{x}{\Set{x,y,z}} = \Set{ x }$, $\divinfo{y}{\Set{x,y,z}} = \Set{ y }$, and $\divinfo{c}{\Set{x,y,z}} = \emptyset$, and, thus, $\divinfoh{t1}{\Set{x,y,z}} = \Set{ x,y }$. In the example \AFW{} in \cref{fig:DivergingPointsp7p8p9} (being the same net as in \cref{fig:example-divpoints-2-p3p5p7}), the set $\Set{p7,p8,p9}$ (colored pink in the figure) has $t4$ (colored blue) and $t1$ (colored orange) as diverging points although $t1$ has no disjoint paths for $p7$ and $p8$ (as it follows the same output).

\begin{figure}[tb]%
	\centering
	\includegraphics[width=0.6\textwidth]{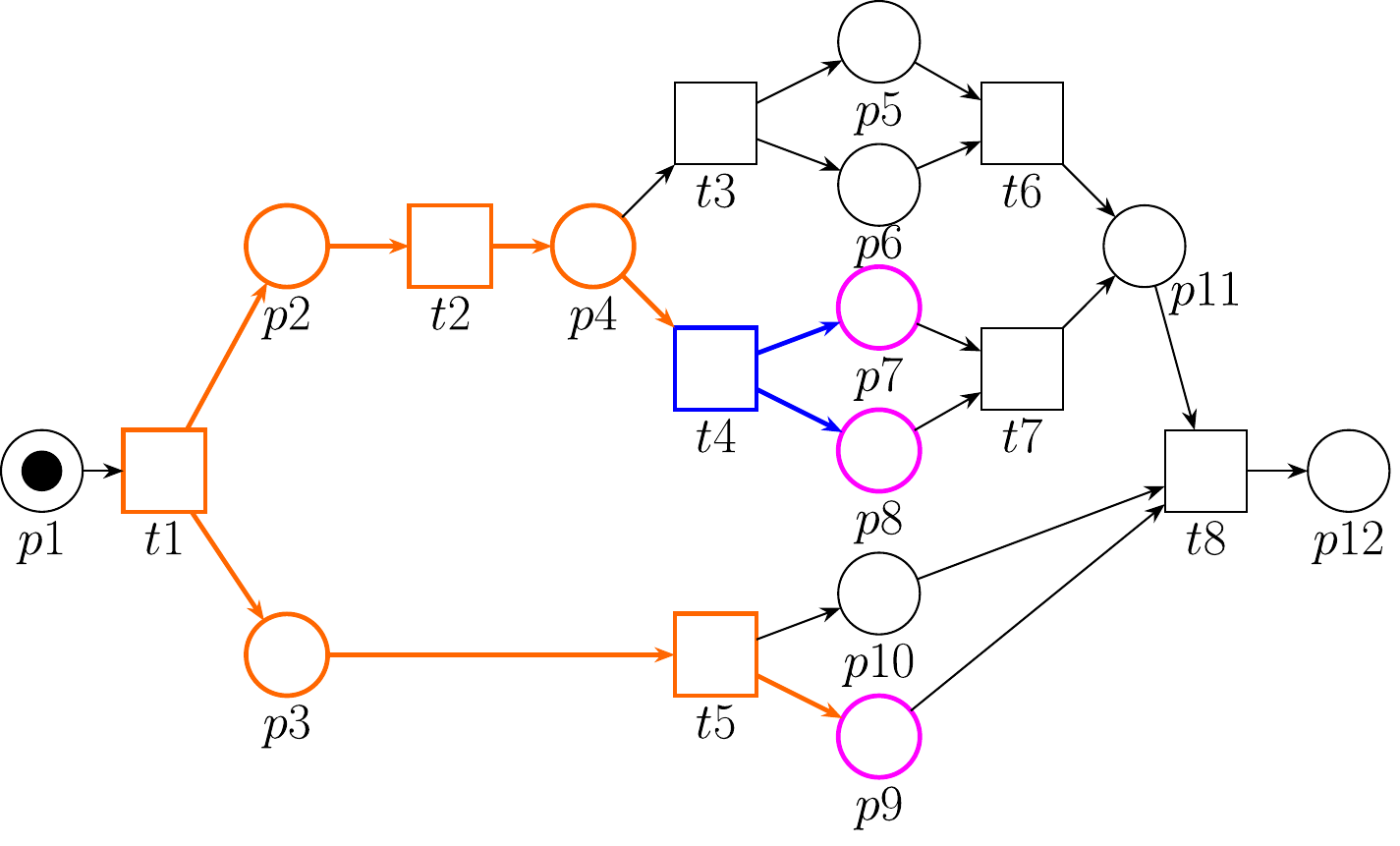}%
	\caption{Diverging points $t_1$ (in orange) and $t_4$ (in blue) for the set of nodes $\Set{p7,p8,p9}$. $t_4$ has blue $(p7)$ and $(p8)$ as disjoint paths to $p7$ and $p8$, respectively; and $t_1$ has orange $(p2,t2,p4,t4,p7)$ and $(p3,t5,p9)$ as disjoint paths to $p7$ and $p9$, respectively.}%
	\label{fig:DivergingPointsp7p8p9}%
\end{figure}

Assume a transition $\delta$ is a diverging point of two concurrent places $x$ and $y$, $\idivpoint{\delta}{\Set{x,y}}$. In this case, tokens can ``travel'' from this transition via disjoint paths to $x$ and $y$, i.\,e., from all markings in which $\delta$ is enabled, there is at least one reachable marking, in which $x$ and $y$ have tokens:

\begin{lemmait}
\label{lemma:IDivPointMarking}
Let $N=(P,T,F,i,o)$ be a sound \AFW{} and $x,y \in P$. It holds: %
\begin{gather*}
	\concurrent{x}{y} \quad \land \quad \delta \in \big( \idivpoints{\Set{x,y}} \cap T \big)\\
	\implies \\
	\forall M_\delta \in \reachableM{N}{\multiset{i}}, \; \postset{\delta} \subseteq M_\delta \; \exists M_{xy} \in \reachableM{N}{M_\delta}\colon \multiset{x,y} \subseteq M_{xy}  \tag*{\defend}
\end{gather*}
\end{lemmait}

\begin{proof}
\proofsize
Constructive proof. %
By $\delta \in \big( \idivpoints{\Set{x,y}} \cap T \big)$ and \cref{def:DivergingPoint}, it follows that there are two disjoint paths $\rho_x \in \Paths{o_x}{x}$ and $\rho_y \in \Paths{o_y}{y}$ from two output places $o_x,o_y \in \postset{\delta}$ to $x$ and $y$, $\pathset{\rho_x} \cap \pathset{\rho_y} = \emptyset$. 
Since $N$ is sound, there is at least one reachable marking $M_\delta \in \reachableM{N}{\multiset{i}}$ with $\postset{\delta} \subseteq M_\delta$ (any marking after firing $\delta$). %
Let $M_\delta$ be such a marking with $o_x,o_y \in M_\delta$. %

$N$ is (simple) free-choice, i.\,e., for each transition $t$ on $\rho_x$ and $\rho_y$, $t \in \big( (\pathset{\rho_x} \cup \pathset{\rho_x}) \cap T \big)$, with $|\preset{t}| \geq 2$, it holds $\forall p \in \preset{t}\colon |\postset{p}| = 1$. %
Since $\pathset{\rho_x} \cap \pathset{\rho_y} = \emptyset$ and since $\forall p \in P\colon |\postset{p}| \geq 2 \implies \preset{(\postset{p})} = \Set{ p }$ (simple free-choiceness), we can force tokens on $\rho_x$ and $\rho_y$ (starting with $o_x$ and $o_y$) to always put a token on a place on $\rho_x$ and $\rho_y$, respectively (i.\,e., transitions on $\rho_x$ and $\rho_y$ fire regardless if they are outputs of diverging places, which could also ``take another path''). %
Since $N$ is sound, there is no dead transition (on $\rho_x$ and $\rho_y$). %
Furthermore, since $\concurrent{x}{y}$ and $N$ is sound there are no paths between $x$ and $y$ and there are no paths between $y$ and $x$ by \cref{lemma:PathAbsence}, i.\,e., no transition on $\rho_x$ and $\rho_y$ may ``require'' the token of $y$ and $x$ to fire, respectively.
\cref{fig:proof-reachable-pair} illustrates the situation. %
For this reason, at least until reaching $x$ and $y$, it is possible that the number of tokens on $\rho_x$ and $\rho_y$ will not decrease. %
Eventually, tokens on $\rho_x$ and $\rho_y$ ``arrive'' at $x$ and $y$ in a marking $M_{xy}$ that is reachable from $M_\delta$, i.\,e., $\forall M_\delta \in \reachableM{N}{\multiset{i}}, \; \postset{\delta} \subseteq M_\delta \; \exists M_{xy} \in \reachableM{N}{M_\delta}\colon \multiset{x,y} \subseteq M_{xy}$. $\checkmark$ \proofend
\end{proof}

\begin{figure}[tb]
	\centering
		\includegraphics[width=0.40\textwidth]{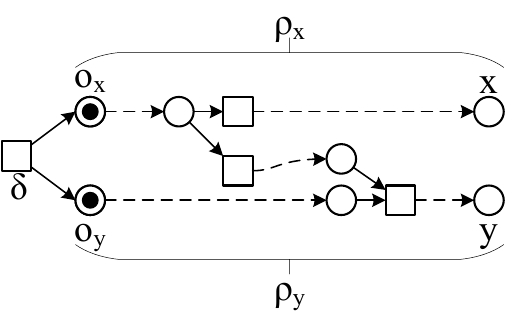}
	\caption{A diverging transition $\delta$ for two places $x$ and $y$ has two output places $o_x$ and $o_y$ with two disjoint paths $\rho_x$ from $o_x$ to $x$ and $\rho_y$ from $o_y$ to $y$, respectively. Through simple free-choiceness, each decision on $\rho_x$ and $\rho_y$ may lead to a remaining token on $\rho_x$ and $\rho_y$, respectively. In addition, soundness ensures that transitions on $\rho_x$ and $\rho_y$, respectively, with multiple input places will be fired. Furthermore, since $\concurrent{x}{y}$ there are no paths between $x$ and $y$ and vice versa by \cref{lemma:PathAbsence}, i.\,e., no transition on $\rho_x$ and $\rho_y$ may require the token of $y$ and $x$, respectively.}
	\label{fig:proof-reachable-pair}
\end{figure}

Assume there is a maximum admissible marking $M_m$ and it is to be checked if it is reachable from the initial marking. If all diverging points $\divpoints{M_m}$ are transitions ($\divpoints{M_m} \subseteq T$), then we can imply that there is a run net, which contains these $\divpoints{M_m}$ transitions with all their output places and the paths between them. Thus, by \cref{cor:ReachabilityWithRunNets}, it directly follows that $M_m$ is reachable. This is the simplest case.  Unfortunately, places can also be diverging points as the above examples have already shown. This complicates the situation since it is not sure which output transition of a diverging place we should add to a run net to check the reachability of $M_m$. Trying each combination of output transitions may again lead to an exponential growth of combinations. Fortunately, soundness and admissibility of $M_m$ simplifies matters as with a sound \AFW{} and an admissible marking, there is always at least one output transition to take. This results from the property that the output transitions of such diverging places must always share the nodes to which they have paths: %
\begin{lemmait}
\label{lemma:SupersetInformation}
Let $N=(P,T,F,i,o)$ be a sound \AFW{} with a marking $M \in \Bag{P}$, a place as diverging point $\delta \in \big(\divpoints{M} \cap P\big)$, and $\Set{ o_x, o_y } \subseteq \postset{\delta}$. It holds: %
\begin{gather}
	M \text{ is admissible} \quad \implies \quad \big( \divinfo{o_x}{\markingset{M}} \subseteq \divinfo{o_y}{\markingset{M}} \big) \lor \big( \divinfo{o_y}{\markingset{M}} \subseteq \divinfo{o_x}{\markingset{M}} \big). \tag*{\defend}
\end{gather}
\end{lemmait}

\begin{proof}
\proofsize
Proof by contradiction. %
\begin{align}
	& M \text{ is admissible} \; \land \; \divinfo{o_x}{\markingset{M}} \not \subseteq \divinfo{o_y}{\markingset{M}} \; \land \; \divinfo{o_x}{\markingset{M}} \not \subseteq \divinfo{o_y}{\markingset{M}} \label{eq:DivCD} \\
	\implies \quad \quad 
	& \exists x \in \divinfo{o_x}{\markingset{M}} \; \exists y \in \divinfo{o_y}{\markingset{M}}\colon \quad x \notin \divinfo{o_y}{\markingset{M}} \land y \notin \divinfo{o_x}{\markingset{M}} \label{eq:DivS2}
\end{align} %
Let $x,y \in M$ be such marked places in $M$. There are exactly two cases: %
\begin{description}
	\item[Case 1:] $\Paths{o_x}{y} \not = \emptyset \lor \Paths{o_y}{x} \not = \emptyset$. %
	Then, however, $y \in \divinfo{o_x}{\markingset{M}}$ or $x \in \divinfo{o_y}{\markingset{M}}$ by \cref{def:DivergingPoint} of diverging points. %
	This contradicts \eqref{eq:DivS2} and this case does not hold. $\lightning$
	\item[Case 2:] $\Paths{o_x}{y} = \emptyset \land \Paths{o_y}{x} = \emptyset$. %
		Thus: %
		\begin{equation}
			\forall \rho_x \in \Paths{o_x}{x} \; \forall \rho_y \in \Paths{o_y}{y}\colon \pathset{\rho_x} \cap \pathset{\rho_y} = \emptyset. \label{eq:DivS3}
		\end{equation} %
		For this reason, the concurrency of $x$ and $y$ cannot ``start'' from $[\delta]$. %
		\cref{fig:proof-superset} illustrates the situation in this case. 
		
		Let $\rho_x \in \Paths{o_x}{x}$ and $\rho_y \in \Paths{o_y}{y}$. %
		By the admissibility of $M$ in \cref{eq:DivCD} and its \cref{def:SufficientMarkings}, it follows $\concurrent{x}{y}$. %
		Therefore, all paths from $x$ and from $y$ to the sink $o$ join at first at converging transitions by \cref{lemma:ConcurrentJointTransition}. %
		Let $c \in T$ be such a converging transition and $\rho_{cx} \in \Paths{x}{c_x}$ with $c_x \in \preset{c}$ is a path without any other converging transition of $x$ and $y$ (i.\,e., $c$ is the ``first'' of such transitions). %
		There is also a disjoint path $\rho_{cy} \in \Paths{y}{c_y}$ with $c_y \in \preset{c}$ by \cref{lemma:ConcurrentJointTransition}, $\pathset{\rho_{cy}} \cap \pathset{\rho_{cx}} = \emptyset$. %
		As (a) there is no path from $o_x$ to $y$, (b) there is no path from $o_y$ to $x$, (c) $c_x$ is the first converging transition on $\rho_{cx}$, and (d) $N$ is sound and simple free-choice ($\preset{o_x} = \preset{o_y} = \Set{ \delta }$), whether $o_x$ or $o_y$ is fired is independent (``free'') from other tokens in any reachable marking $M$ of $N$.
		For this reason and $N$ is sound, let $M_\delta \in \reachableM{N}{\multiset{i}}$ and $M_{c_y,\delta} \in \reachableM{N}{M_\delta}$ be two reachable markings with $\delta \in M_\delta$ and $\multiset{c_y, \delta} \subseteq M_{c_y,\delta}$ so that $c$ can be fired if $\delta$ decides for $o_x$. %
		However, in any $M_{c_y,\delta}$, there is an ``unsafe'' path from $\delta$ via $c_y$ to the sink $o$ by \cref{theorem:PathToEnd}. %
		This contradicts \eqref{eq:DivCD} that $N$ is sound. $\lightning$
		This case does not hold. %
\end{description}	%
	
	\noindent Both cases do not hold. %
	Thus, assumption~\eqref{eq:DivCD} does not hold. $\lightning$ %
	Thus, the lemma must hold. \proofend
\end{proof}

\begin{figure}[tb]
	\centering
		\includegraphics[width=0.50\textwidth]{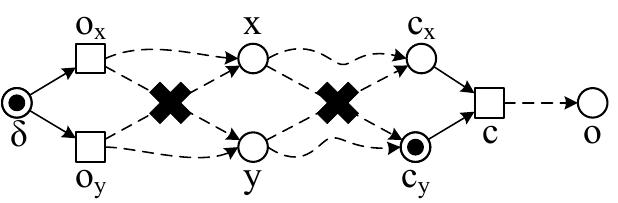}
	\caption{A diverging place $\delta$ has two output transitions $o_x$ and $o_y$. $o_x$ has a path to $x$ but not to $y$ and $o_y$ has a path to $y$ but not to $x$. By \concurrent{x}{y} and \cref{lemma:ConcurrentJointTransition}, there is a ``first'' joining transition $c$. Thus, there is no path from $y$ to $c_x$ and there is no path from $x$ to $c_y$. If $\delta$ decides for $o_x$, no token ``after'' $\delta$ can reach to $c_y$ enabling $c$. Thus, since soundness is required, there must be a marking $M_{c_y,\delta}$ with $\multiset{c_y, \delta} \subseteq M_{c_y,\delta}$ to avoid a dead $c$. However, the path from $\delta$ via $o_y$, $c_y$, and $c$ to $o$ is not safe.}
	\label{fig:proof-superset}
\end{figure}

The above \cref{lemma:SupersetInformation} guarantees that the output transitions $\postset{\delta}$ of a diverging place $\delta$ of an admissible marking $M_m$ either have paths to the same subset of places of $M_m$ or at least one output transition has paths to more places of $M_m$. To simplify what matters: there is always one output transition of a diverging place that has paths to the union of all subsets of $M_m$ of all other output transitions. Since we are interested in the (sub-marking) reachability of all places of $M_m$, it is always the best option to take that output transition with the most reachable places of $M_m$. In summary, if we derive $\divpoints{M_m}$ for an admissible marking $M_m$, we could take those output transitions of diverging places containing the greatest subset of places of $M_m$.

The remaining open question is if there is at least one transition in $\divpoints{M_m}$, which has paths to all places in $M_m$. By \cref{def:DivergingPoint}, such a transition $\delta \in \divpoints{M_m}$ exists if $\divinfoh{\delta}{\markingset{M_m}} = M_m$ (remember that $\divinfoh{\delta}{\markingset{M_m}} \coloneqq \bigcup_{o_\delta \in \postset{\delta}} \divinfo{o_\delta}{\markingset{M_m}}$). If such $\delta$ exists, then we can imply a partial run net starting by $\delta$ and adding all transitions and necessary paths to these transitions as well as to the places of $M_m$ and the necessary paths to those places. Although this partial run net may not be a ``full'' run net, it implies the existence of a ``full'' run net.

\begin{figure}[t]
	\centering
		\includegraphics[width=0.80\textwidth]{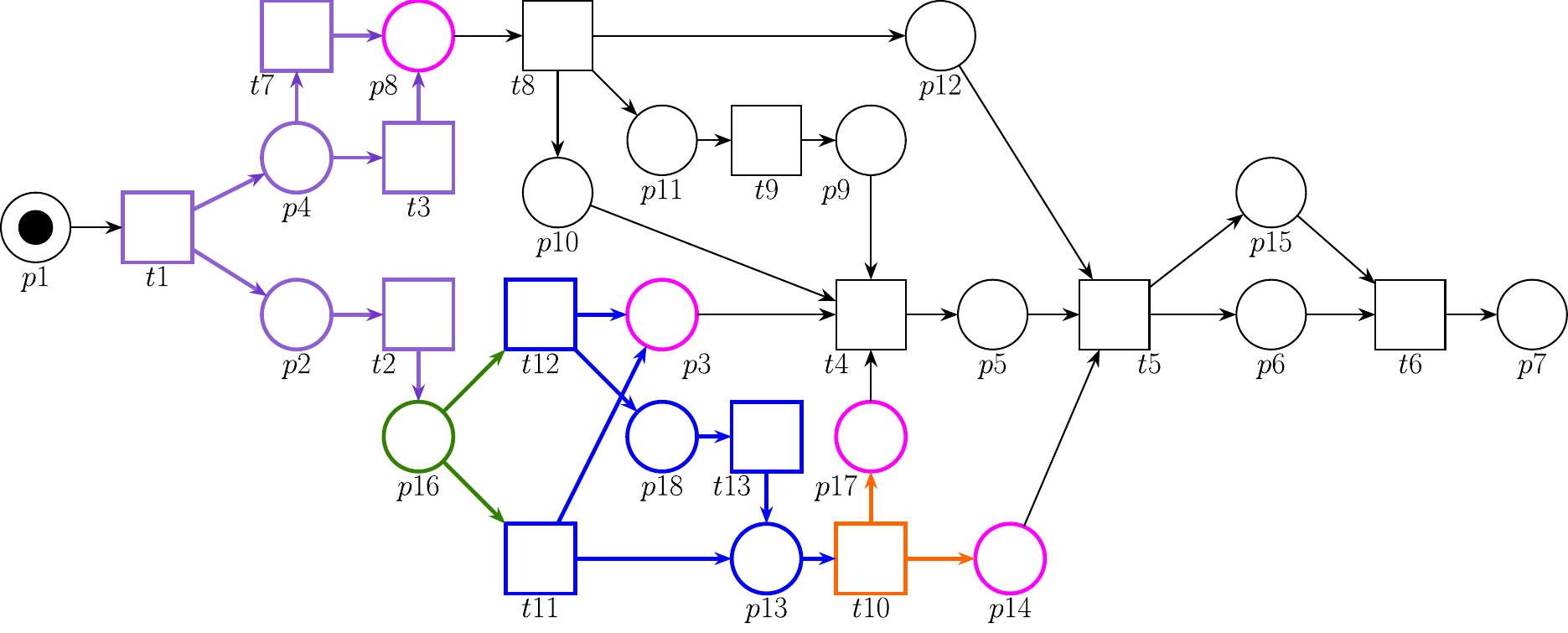}
		\vspace{-0.35cm}
	\caption{The net of \cref{fig:example} investigating the maximum admissible marking $\multiset{p3,p8,p14,p17}$ (colored in pink). The diverging points of this marking are $t10$ (colored in orange), $t11$ (blue), $t12$ (blue), $p16$ (green), and $t1$ (violet).}
	\label{fig:examplep3p8p14p17}
\end{figure}

Let us consider our example in \cref{fig:example} being colored in \cref{fig:examplep3p8p14p17}. Assume that we want to check whether the marking $M_m = \multiset{p3,p8,p14,p17}$ is reachable, or not (colored in pink in the figure). All places are pairwise concurrent and it is not possible to add any further place ($M_m$ is maximum admissible). For this reason, we have to check whether there is a diverging transition which can cause this marking. The diverging points for $M_m = \multiset{p3,p8,p14,p17}$ cover the set $\Set{t10,t11,t12,p16,t1}$ by \cref{def:DivergingPoint}. $t10$'s output places $p14$ and $p17$ (colored orange/pink in the figure) are part of $M_m$, thus, $\divinfo{p14}{\markingset{M_m}} = \Set{ p14 }$ and $\divinfo{p17}{\markingset{M_m}} = \Set{ p17 }$. For $t12$'s output places $p3$ and $p18$ (colored blue/pink in the figure), it holds $\divinfo{p3}{\markingset{M_m}} = \Set{ p3 }$ and $\divinfo{p18}{\markingset{M_m}} = \Set{ p14, p17 }$. For $t11$ (also colored blue/pink), the information are similar $\divinfo{p3}{\markingset{M_m}} = \Set{ p3 }$ and $\divinfo{p13}{\markingset{M_m}} = \Set{p14,p17}$. The place $p16$ (colored green) is a decision with two output transitions $t11$ and $t12$. By \cref{lemma:SupersetInformation}, both transitions either have the same information ($\divinfo{t11}{\markingset{M_m}} = \divinfo{t12}{\markingset{M_m}} = \Set{ p3,p14,p17 }$) or one is a superset of the other. Eventually, transition $t1$ (colored violet in the figure) has the output places $p4$ with $\divinfo{p4}{\markingset{M_m}} = \Set{p8}$ and $p2$ with $\divinfo{p2}{\markingset{M_m}} = \Set{ p3,p14,p17 }$. Now, we can induce a partial run net from this information. \Cref{fig:example-partial} illustrates the induction steps. First, we can add the diverging point $t1$, which is a diverging point of all places of $M_m$. Furthermore, we add its output places $p2$ and $p4$ as well as a path from $p1$ (the source) to $t1$. We further add $p8$ as a place of $M_m$ and a path to it as we know that $p4$ has a disjoint path to $p8$. The diverging place $p16$ with a path to it is added next (all paths from $t1$ to $t11$ or $t12$ pass $p16$). Furthermore, we add that output transition with the most information (i.\,e., that output transition $t$ for which $\divinfo{t}{\markingset{M}}$ is the greatest). In this case, it does not matter, which of $t11$ or $t12$ we add as they have the same information. In the next step, we add the diverging transition $t12$ (which we already have inserted) with its output places $p3$ and $p18$. Eventually, diverging transition $t10$ is added with a path to it as well as its output places. This final partial run net contains all places of $M_m = \multiset{p3,p8,p14,p17}$ and implies the existence of a ``full'' run net containing all places of $M_m$. Such steps are applied to finally deciding the (sub-marking) reachability in \AFWs{}:

\begin{figure}
	\centering
	\begin{minipage}[t]{0.48\textwidth}
		\begin{center}
			\includegraphics[width=0.50\textwidth]{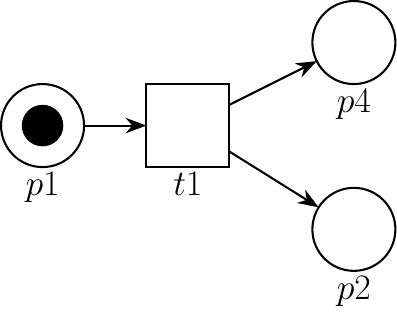}
		\end{center}
		Adding diverging point $t1$ with its output places $p2$ and $p4$.
	\end{minipage}%
	\hfill%
	\begin{minipage}[t]{0.48\textwidth}
		\begin{center}
			\includegraphics[width=0.50\textwidth]{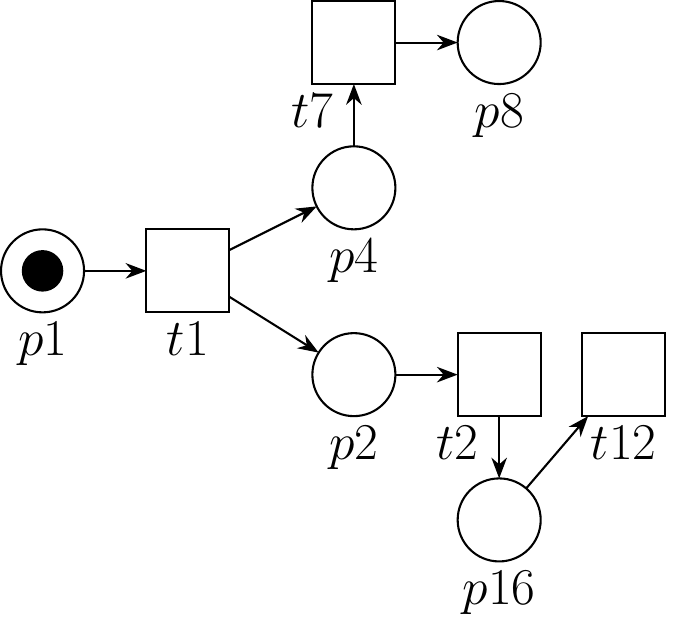}
		\end{center}
		Adding a path from $p4$ to $p8$. Furthermore, adding diverging point $p16$ with its output transition $t12$. It does not matter, which of $t11$ or $t12$ is added as they have the same information ($\divinfo{t11}{\markingset{M_m}} = \divinfo{t12}{\markingset{M_m}}$).
	\end{minipage}\\\vspace{0.5cm}
	\begin{minipage}[t]{0.48\textwidth}
		\begin{center}
			\includegraphics[width=0.75\textwidth]{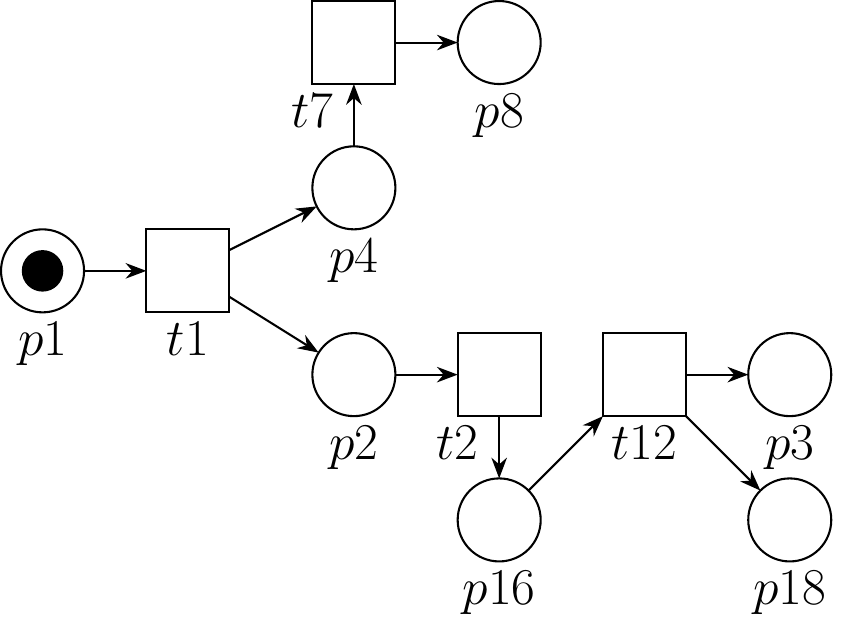}
		\end{center}
		Adding diverging point $t12$ with its output places $p3$ and $p18$.
	\end{minipage}%
	\hfill%
	\begin{minipage}[t]{0.48\textwidth}
		\begin{center}
			\includegraphics[width=1.00\textwidth]{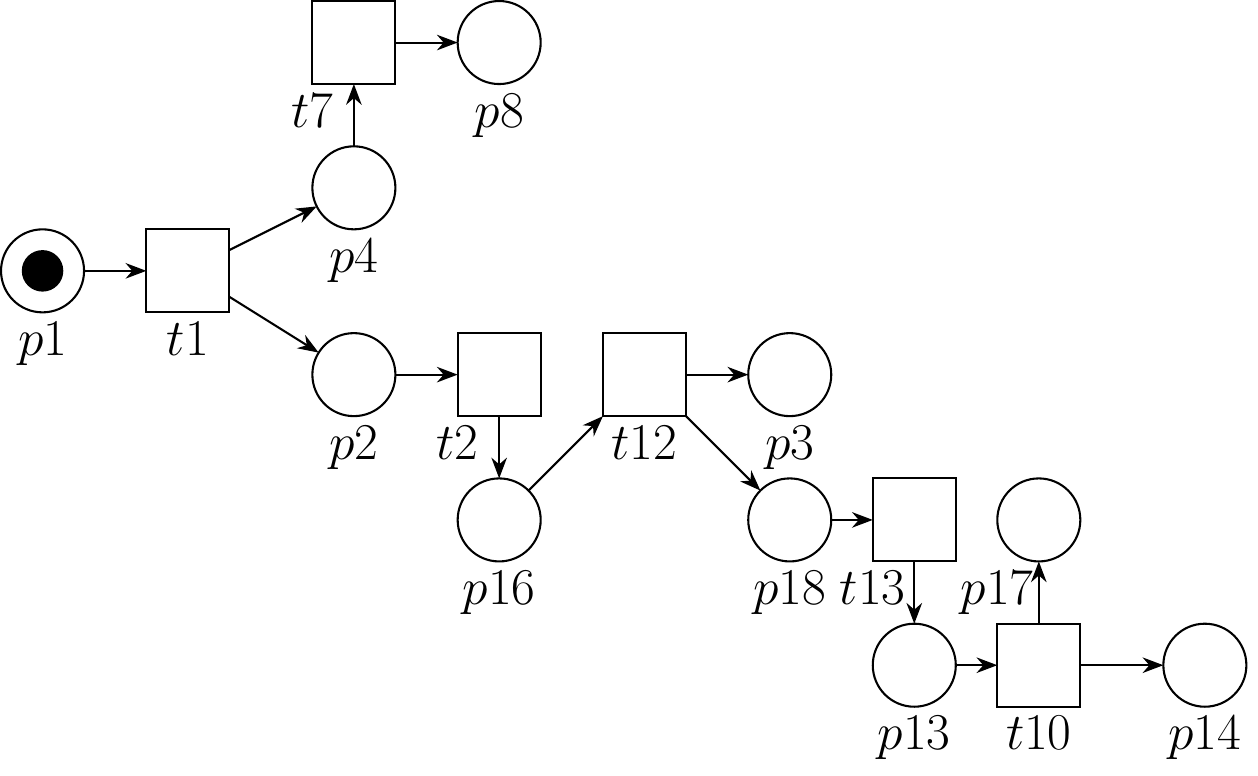}
		\end{center}
		Adding a path to the diverging point $t10$ with its output places $p14$ and $p17$. This net contains all places $\multiset{p3,p8,p14,p17}$ and \ldots
	\end{minipage}\\\vspace{0.5cm}
	\begin{minipage}[t]{1.00\textwidth}
		\begin{center}
			\includegraphics[width=0.80\textwidth]{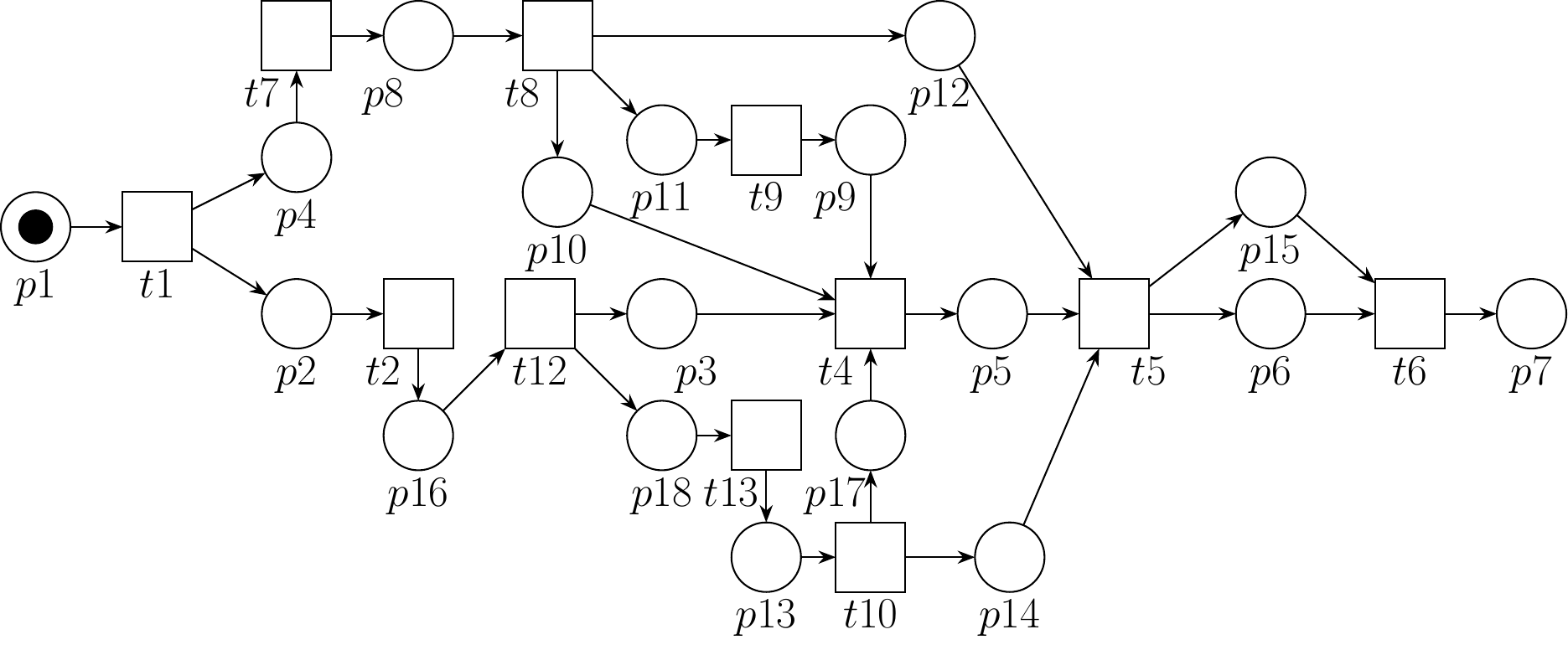}\\
			\ldots implies the existence of a run net.
		\end{center}		
	\end{minipage}%
	\caption{Step-wisely inducing a partial run net for the exemplary \AFW{} in \cref{fig:example} out of the diverging point information for the marking $M_m = \multiset{p3,p8,p14,p17}$. The final partial run net implies the existence of a run net presented at the end.}
	\label{fig:example-partial}
\end{figure}

\begin{theoremit}[(Sub-Marking) Reachability]
\label{theorem:Reachability}
Let $N=(P,T,F,i,o)$ be a sound \AFW{} with a marking $M \in \Bag{P}$. It holds: %
\begin{align}
	M \in \reachableM{N}{\multiset{i}} \quad \iff \quad & M \text{ is maximum admissible } \quad \land \label{eq:R1} \\ & \exists \delta \in (\divpoints{M} \cap T)\colon \divinfoh{\delta}{\markingset{M}} = M 
	\nonumber \\ %
	\exists M_s \in \reachableM{N}{\multiset{i}}\colon M \subseteq M_s \quad \iff \quad & M \text{ is admissible } \quad \land  \label{eq:R2} \\ & \exists \delta \in (\divpoints{M} \cap T)\colon \divinfoh{\delta}{\markingset{M}} = M 
	\tag*{\defend}
\end{align}
\end{theoremit}

\begin{proof}
\proofsize
Constructive proof. %
By \cref{cor:ReachabilityWithRunNets} it holds: %
\begin{equation*}
	M \in \reachableM{N}{\multiset{i}} \iff M \text{ is maximum admissible} \; \land \; \exists (P_\sigma,T_\sigma,F_\sigma) \in \RunNets{N}{\multiset{i}}\colon \markingset{M} \subseteq P_\sigma.
\end{equation*} %
Furthermore, by \cref{lemma:SupersetInformation}, it follows for each diverging transition $\delta \in \divpoints{M} \cap T$ and for each of its output places $o_\delta \in \postset{\delta}$ that there is at least one sub-net of all $\divinfo{o_\delta}{\markingset{M}}$ from $o_\delta$ to all $x \in \divinfo{o_\delta}{\markingset{M}}$, which diverges only in transitions. %
As a consequence, for each diverging transition $\delta \in \divpoints{M} \cap T$, the union $\bigcup_{o_\delta \in \postset{\delta}} \divinfo{o_\delta}{\markingset{M}} = \divinfoh{\delta}{\markingset{M}}$ describes the subset $M' \subseteq M$ of $M$ to which the transition $\delta$ has such a sub-net, which only diverges in transitions. %
If there is no diverging transition $\delta \in \divpoints{M} \cap T$ with $M = \divinfoh{\delta}{\markingset{M}}$, then there is no run net that can contain a sub-net, which only diverges at transitions to the places of $M$. %
Therefore, there is no run net that contains all places of $M$. %
If there is such a diverging transition $\delta$ with $M = \divinfoh{\delta}{\markingset{M}}$, then there must be at least one run net, which contains all marked places in $M$. $\checkmark$ %
The proof of sub-marking reachability by \eqref{eq:R2} follows the same argumentation. $\checkmark$ \proofend %
\end{proof}

With \cref{theorem:Reachability}, we can decide if a (sub) marking is reachable based on the structure of a sound \AFW{} rather than over its concrete behavior in terms of occurrence sequences, state space exploration, or similar techniques. In addition, maximum admissibility and diverging points help to understand why a specified marking is \emph{not} reachable. For example, the marking $\multiset{p8,p5,p14}$ in \cref{fig:example} is \emph{not} reachable from $\multiset{p1}$ since $p8$ is not concurrent to $p5$ (there is a path from $p8$ to $p5$ following \cref{lemma:PathAbsence}). In \cref{fig:MaxNotReach}, the marking $\multiset{x,y,z}$ is not reachable from $\multiset{i}$ as there is no diverging transition which leads to simultaneous tokens on $x$, $y$, and $z$ --- just on $x$ and $y$, $x$ and $z$, or $y$ and $z$. This diagnostic information may support repairing improperly developed nets, in which a desired marking is not reachable or an undesired marking is reachable. In \cref{subsec:DiagnosticsAll} will be discussed how the diagnostics of the algorithm introduced in the following can be used.

\subsection{Algorithmic Derivation}

There is a related concept in compiler theory to identify diverging points: the placement of so-called {$\phi$-functions} to derive the minimal \emph{Static Single Assignment} (SSA) form \cite{DBLP:journals/toplas/CytronFRWZ91}. The SSA form requires that each \emph{program variable} in a computer program gets statically once a value. Transforming a program into SSA form introduces no difficulties as long as the program does not have conditional statements or loops (i.\,e., it is a sequence of statements). Otherwise, different values of the same variable can be valid at some point in the program. See \cref{fig:SSAform} (a) as an example. It shows a control-flow graph with variables $v$ and $a$ whereas $v$ has multiple assignments with different values, thus, it is not in a SSA form. In \cref{fig:SSAform} (b), each assignment to $v$ is replaced with a unique variable $v_1$, $v_2$, \ldots, $v_9$ called \emph{definition} of $v$. The difficulty occurs when different definitions of the same variable reach the same node. For example in \cref{fig:SSAform} (b), either definition $v_1$ or definition $v_3$ is valid in the dashed-lined node depending on the actual execution at runtime (the upper path via the node defining $v_3$ or the ``middle'' path without an additional definition of $v$). This would hinder to transform every program into SSA form. Thus, if there are converging control-flows with multiple definitions of the same variable, a \emph{$\phi$-function} selects that definition of the variable, which was used during execution and assigns it to a new definition of the variable. For example, definition $v_4$ in \cref{fig:SSAform} (b) has $\phi(v_3,v_1)$ as an assignment meaning that either definition $v_3$ for the left incoming flow or definition $v_1$ for the lower incoming flow will be taken. Note that the order of the arguments of the $\phi$-function depends on the order of the incoming flows and, therefore, is not random. This implies that at most one definition can be reached via each incoming flow.

\begin{figure}[t]
	\begin{minipage}[t]{0.48\textwidth}
		\begin{center}
			\includegraphics[width=1.0\textwidth]{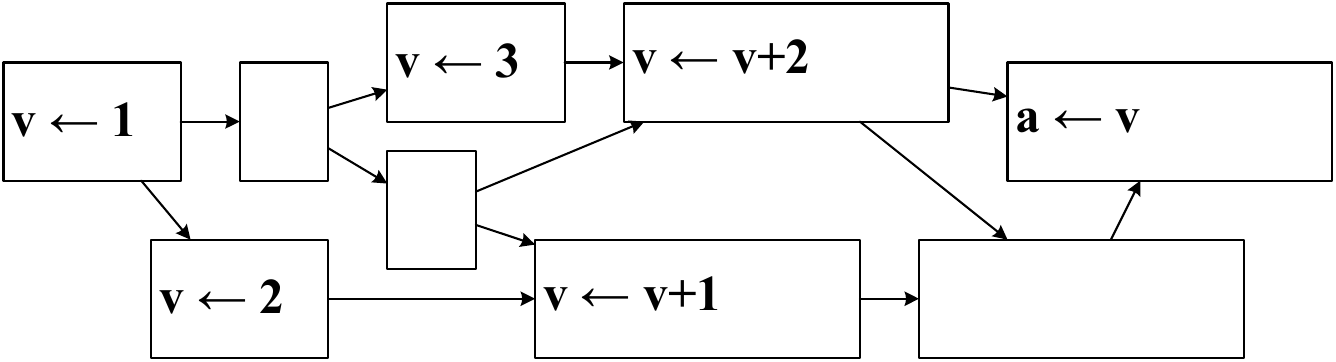}\\
			(a)
		\end{center}
	\end{minipage} \hfill %
	\begin{minipage}[t]{0.48\textwidth}
		\begin{center}
			\includegraphics[width=1.0\textwidth]{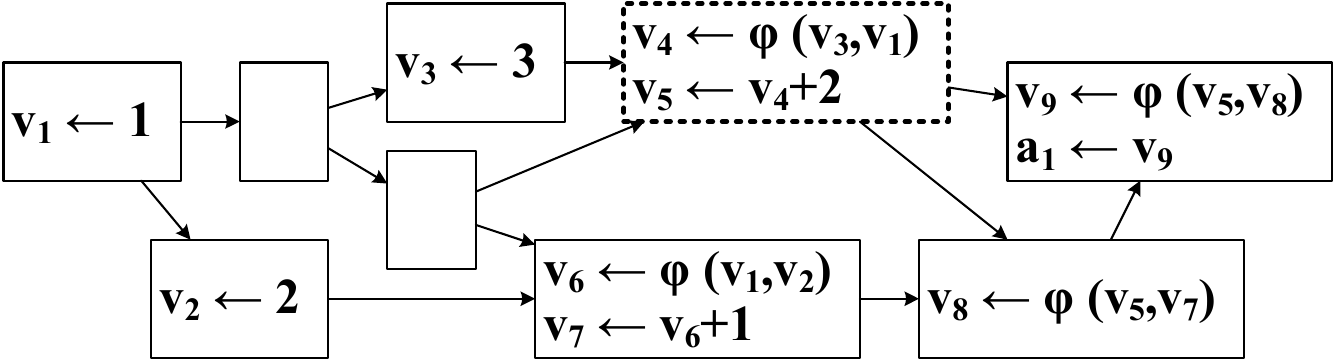}\\
			(b)
		\end{center}
	\end{minipage}	
	\caption{Transforming a program with program variables $v$ and $a$ (a) into SSA form (b) by inserting new definitions of variables and $\phi$-functions at converging nodes.}
	\label{fig:SSAform}
\end{figure}

Cytron et al.\@ \cite{DBLP:journals/toplas/CytronFRWZ91} derived algorithms to compute the \emph{minimal} SSA form from normal programs. The \emph{minimal} SSA form only inserts $\phi$-functions if necessary. To achieve a SSA form out of a computer program, the transformation first inserts a new definition of a program variable for each assignment of this variable as explained previously. Subsequently, it inserts a $\phi$-function of the form $d_n = \phi(d_1,\ldots,d_m)$, $m \geq 2$, for each program variable at each program point where different definitions $d_1,\ldots,d_m$ can occur. This placement of $\phi$-functions in the minimal SSA assumes fairness, i.\,e., each path can be followed by a token. The positions of nodes where a $\phi$-function is required for a program variable $d$ are positions where two paths from two different definitions of the same variable first meet \cite{DBLP:journals/toplas/CytronFRWZ91}. As a consequence, for all $d_1,\ldots,d_m$ there are disjoint paths from the position of the $\phi$-function to their assignment. Actually, this is the same definition as for diverging points, only in the reverse direction. For this reason, we will reuse the algorithms for identifying diverging points.

For placing $\phi$-functions, the \emph{post-dominance relation} \cite{DBLP:journals/siamcomp/Tarjan74} plays an important role, in which a node $x$ \emph{post-dominates} $y$ if $x$ appears on any path from $y$ to a single end node. Resulting from the post-dominance relation, there is a boundary where the post-dominance of a node $x$ ``ends'', i.\,e., $x$ post-dominates a direct successor (output) of a node $y$ but does not post-dominate $y$ itself. This boundary is called the \emph{post-dominance frontier}. Positions, where two paths to different nodes first diverge, are closely related to the \emph{post-dominance frontier} \cite{DBLP:journals/toplas/CytronFRWZ91}:

\begin{definition}[Post-Dominance]
\label{def:Dominance}
Let $N=(P,T,F,i,o)$ be a workflow net with two nodes $x,y \in P \cup T$. %
$x$ \emph{post-dominates} $y$, denoted by $\pdomi{x}{y}$, if $x$ is on all paths from $y$ to $o$: %
\[ \pdomi{}{y} \coloneqq \Set{ x \in P \cup T \given \forall \rho \in \Paths{y}{o}\colon x \in \pathset{\rho} }. \]%
$x$ \emph{strictly} post-dominates $y$, denoted by $\spdomi{x}{y}$, if $\pdomi{x}{y}$ and $x \not = y$. %
$x$ \emph{immediately} post-dominates $y$, depicted $\ispdomi{x}{y}$, if $\spdomi{x}{y}$ and $\forall \spdomi{z}{y}\colon \notspdomi{x}{z}$ \cite{DBLP:journals/siamcomp/Tarjan74}. %
The \emph{post-dominance frontier} of $x$, depicted $\pdomfront{x}$, is the set of all nodes $Y \subseteq P \cup T$ so that for any $y \in Y$ there is an output of $y$ being post-dominated by $x$, but $x$ does not strictly post-dominate $y$: %
\[ \pdomfront{x} \coloneqq \Set{ y \in P \cup T \given \exists s \in \postset{y}\colon \pdomi{x}{s} \land \notspdomi{x}{y}}. \]
The extension of the post-dominance frontier to a set of nodes $X$ is: %
\[ \pdomfront{X} \coloneqq \bigcup_{x \in X} \pdomfront{x}. \tag*{\defend} \] %
\end{definition} %
Since the \emph{immediate} post-dominance relation builds a tree, each node has exactly one immediate post-dominating node except the sink, but can immediately post-dominate more than one node \cite{DBLP:journals/siamcomp/Tarjan74}.

In \cref{fig:exampleAlgo} (showing \cref{fig:example}), the sink $p7$ naturally post-dominates all other nodes of the net inclusively itself. $p13$ strictly post-dominates $\Set{p18,t13}$ (but not $t11$ and $t12$) and immediately post-dominates $t13$. The post-dominance frontier $\pdomfront{p13}$ of $p13$ contains $\Set{t11,t12}$. Since $\pdomfront{p5} = \Set{t8,t10,t12}$, the combined post-dominance frontier of $p13$ and $p5$ is: %
\[ \pdomfront{\Set{p5,p13}} = \Set{t8,t10,t11,t12}. \]

Following the thoughts of Cytron et al.\@ \cite{DBLP:journals/toplas/CytronFRWZ91} for the reverse net in our case, the set of all diverging points $\divpoints{D}$ of a set of nodes $D$ is related to the \emph{iterated} post-dominance frontier $\itpdomfront{D}$. $\itpdomfront{D}$ is the limit of the increasing sequence of node sets: %
\begin{align*}
	\mathit{PDF}_1 = & \; \pdomfront{D}; \\
	\mathit{PDF}_{i + 1} = & \; \pdomfront{D \cup \mathit{PDF}_i}, \\%
	\text{so that eventually follows:}\\ %
	\itpdomfront{D} \supseteq & \; \divpoints{D}.%
\end{align*}
In the example above, we considered the places $p5$ and $p13$ with $\pdomfront{\Set{p5,p13}} = \Set{t8,t10,t11,$ $t12}$. The iterated post-dominance frontier will then be defined iteratively: %
\begin{align*}
	\mathit{PDF}_1 = & \; \pdomfront{\Set{p5,p13}} = \Set{t8,t10,t11,t12}; \\ %
	\mathit{PDF}_2 = & \; \pdomfront{\Set{p5,p13} \cup \mathit{PDF}_1} = \pdomfront{\Set{p5,p13,t8,t10,t11,t12}} \\ %
	               = & \; \Set{t8,t10,t12} \cup \Set{t11,t12} \cup \Set{t1} \cup \Set{t11,t12} \cup \Set{p16} \cup \Set{p16} \\ %
								 = & \; \Set{p16,t1,t8,10,t11,t12}; \\ %
	\mathit{PDF}_3 = & \; \pdomfront{\Set{p5,p13} \cup \mathit{PDF}_2} = \pdomfront{\Set{p5,p13,p16,t1,t8,10,t11,t12}} \\ %
	               = & \; \Set{t8,t10,t12} \cup \Set{t11,t12} \cup \Set{t1} \cup \Set{p1} \cup \Set{t1} \cup \Set{t11,t12} \cup \Set{p16} \cup \Set{p16} \\ %
								 = & \; \Set{p1,p16,t1,t8,10,t11,t12} \\ %
								 & \; \text{Since $\pdomfront{p1} = \emptyset$ as new node will not change the iterative post-dominance frontier:} \\
								 = & \; \itpdomfront{\Set{p5,p13}}
\end{align*} %
As the reader can verify, $\itpdomfront{\Set{p5,p13}} = \Set{p1,p16,t1,t8,10,t11,t12}$ contains all diverging points of $D = \Set{p5,p13}$. $p1$ is not a diverging point as 
it only has one output ($t1$). Thus, $\itpdomfront{\Set{p5,p13}}$ is a superset of $\divpoints{\Set{p5,p13}}$ as stated above.

\begin{figure}[t]
	\centering
		\includegraphics[width=0.8\textwidth]{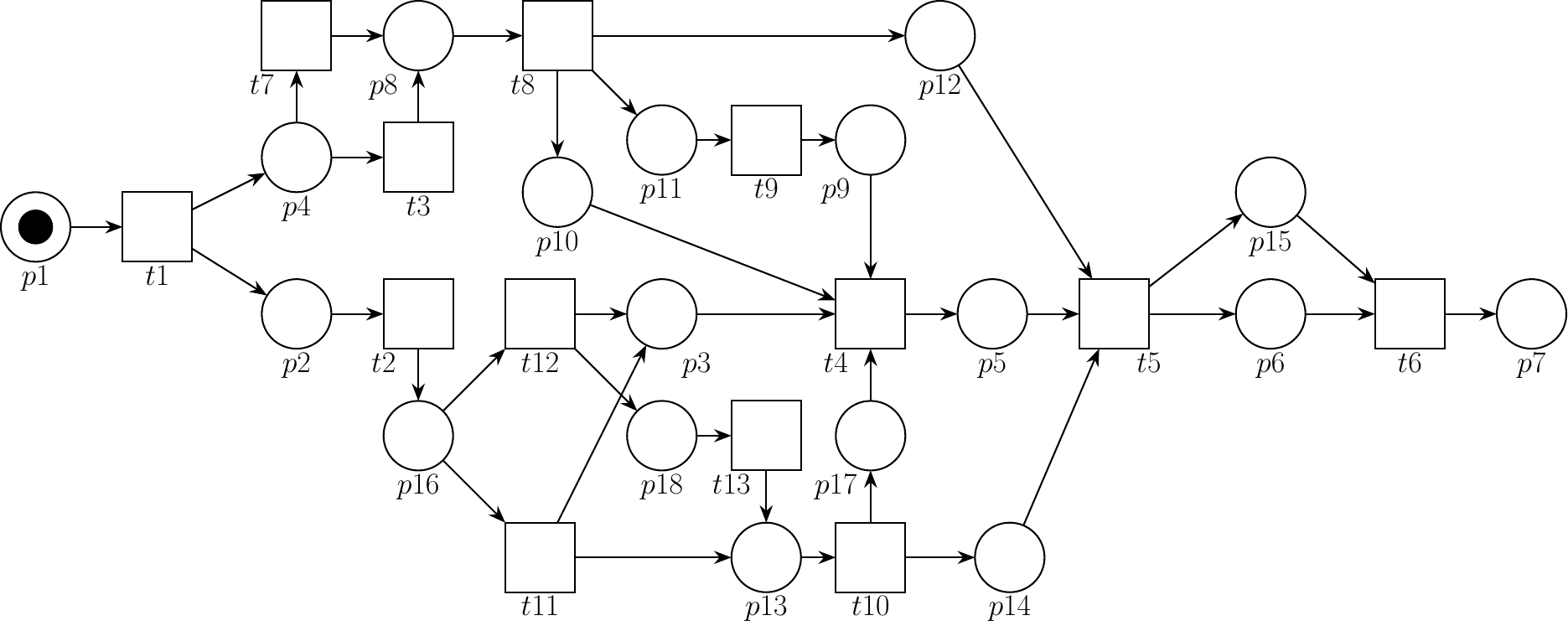}
		\vspace{-0.35cm}
	\caption{The sound acyclic free-choice workflow net of \cref{fig:example}.}
	\label{fig:exampleAlgo}
\end{figure}

We can simply reuse the algorithm of computing the post-dominance frontier to identify all diverging points $\divpoints{M}$ of a given admissible marking $M$. Instead of using the algorithm provided by Cytron et al.\@~\cite{DBLP:journals/toplas/CytronFRWZ91}, the ``engineered'' algorithm of Cooper et al.\@~\cite{Cooper} is much simpler, only computes necessary information, and is widely used in compiler construction. Since the algorithm of Cooper et al.\@ computes the \emph{dominance frontier} instead of the \emph{post}-dominance frontier, the here stated algorithms are modified accordingly. %
In the following, we will discuss the algorithms required to compute the diverging points by presenting first \cref{algo:ImmediatePostDominance} to compute the immediate post-dominance relation. This is followed by \cref{algo:PostDominanceFrontier} to compute the post-dominance frontier. Both algorithms are used to compute the diverging points of a given marking in \cref{algo:DivergingPoints}. Finally, \cref{algo:ReachabilityAcyclic} requires these diverging points to decide whether a given marking is reachable, or not.

\begin{algorithm}[t]
\caption{Computing the immediate post-dominance relation for a workflow net $N=(P,T,F,i,o)$ following \cite{Cooper}.}
\label{algo:ImmediatePostDominance}
\SetKwProg{Fn}{Function}{}{end}
\Fn{computeImmediatePostDoms($N$)}{
	\lFor{$x \in P \cup T$}{
		$\ispdomi{}{x} \gets \text{undefined}$
	}
	$\ispdomi{}{o} \gets o$\;
	\tcp{Compute the rev.\!\!\! topolog.\!\!\! order of the nodes with their numbers.}
	$\mathit{order}, \mathit{ordernr} \gets$ revTopologicalOrder($N$)\; 
	$\mathit{changed} \gets \mathit{true}$\;
	\While{$\mathit{changed}$}{
		$\mathit{changed} \gets \mathit{false}$\;
		\For{$x \in \mathit{order}$}{\nllabel{al:ip:or}
			\If{$\postset{x} \not = \emptyset$}{
				$\mathit{new\_ipdom} \gets \text{first}(\postset{x})$\;\nllabel{al:ip:first}
				\For{$s \in \postset{x}$}{
					\If{$\ispdomi{}{x} \not = \text{\normalfont{undefined}} \land s \not = \text{\normalfont{first}}(\postset{x})$}{
							$\mathit{new\_ipdom} \gets \text{comImmPostDom}(s, \mathit{new\_ipdom}, \ispdomi{}{}, \mathit{ordernr})$\;\nllabel{al:ip:others}
					}
				}
				\If{$\ispdomi{}{x} \not = \mathit{new\_ipdom}$}{
					$\ispdomi{}{x} \gets \mathit{new\_ipdom}$\; \nllabel{al:ip:up1}
					$\mathit{changed} \gets \mathit{true}$\;
				}
			}
		}
	}
	\KwRet{$\ispdomi{}{}$}
}
\end{algorithm}

The idea of Cooper et al.\@ is to directly compute the \emph{immediate} post-dominance relation without the necessity of previously computing the post-dominance relation. The algorithm is presented in \cref{algo:ImmediatePostDominance}. One of the main characteristics is that the algorithm works on the nodes in a reverse topological order (\cref{al:ip:or}) such that a node is only visited if all its outputs were already investigated. The simple thought behind this order is that a node $n$ being not reachable from any output of $n$ cannot (immediately) post-dominate $n$. 

\Cref{al:ip:first} of \cref{algo:ImmediatePostDominance} takes the first output $o_1$ of the current node $x$ as immediate post-dominator $\mathit{new\_ipdom}$. If $|\postset{x}| = 1$, then $\mathit{new\_ipdom}$ is correctly set. Otherwise, $\mathit{new\_ipdom}$ is investigated against all other outputs ($\not = o_1$) in \cref{al:ip:others}. This uses the \emph{comImmPostDom} function presented in \cref{algo:Intersect}. Although looking complicated, the principle of this \emph{comImmPostDom} function is simple as explained on an example: Take \cref{fig:exampleAlgo} with $x = t10$ as the current node. We know that all nodes ``after'' $t10$ were already investigated (i.\,e., \Set{p7,t6,p6,p15,t5,p14,p5,t4,p17} in reverse topological order). \cref{al:ip:first} of \cref{algo:ImmediatePostDominance} takes $\mathit{new\_ipdom} = p17$ of its outputs $\Set{p14,p17}$ (for example). Thus, $p14$ is the remaining output and \emph{comImmPostDom} of \cref{algo:Intersect} is invoked with $p17$ and $p14$ as arguments. Now, imagine you use your fingers: one on $p17$ ($\mathit{fingerX}$) and one on $p14$ ($\mathit{fingerY}$). Both fingers are not on the same node, thus, at least one of both is not an immediate post-dominator of $x$. Since $p17$ appears \emph{after} $p14$ in the order for this example, $p17$'s ordering number is greater than that of $p14$. Thus, we begin to slide with the finger on $p14$ ($\mathit{fingerY}$) to the immediate post-dominator of $p14$, which is $t5$ (i.\,e., $\mathit{fingerY}$ is $t5$ now). $t5$ appears \emph{before} $p17$ in the order, thus, having a smaller number than $p17$ (and than $t4$ and $p5$). For this reason, we stop sliding with $\mathit{fingerY}$ and start to slide with $\mathit{fingerX}$ from one immediate post-dominator to the next (i.\,e., from $p17$ via $t4$ and $p5$ to $t5$). We reach $t5$ with both fingers, therefore, stopping all \emph{while}-loops in \cref{algo:Intersect}. $t5$ is returned as the new (and correct) immediate post-dominator of $t10$. Since $x = t10$ does not have further output nodes, the immediate post-dominator of $x = t10$ is updated accordingly in \cref{al:ip:up1}. Once all nodes were processed, a relation (as a map) is returned that contains for all nodes their immediate post-dominators.

\begin{algorithm}[t]
\caption{Finding the closest immediate post-dominator of two nodes $\mathit{fingerX}$ and $\mathit{fingerY}$ given the current immediate post-dominance relation $\ispdomi{}{}$ and ordering numbers $\mathit{ordernr}$ in the post-order following \cite{Cooper}.}
\label{algo:Intersect}
\SetKwProg{Fn}{Function}{}{end}
\Fn{comImmPostDom($\mathit{fingerX}$, $\mathit{fingerY}$, $\ispdomi{}{}$, $\mathit{ordernr}$)}{
	\While{$fingerX \not = fingerY$}{
		\While{$\mathit{ordernr}(fingerX) < \mathit{ordernr}(fingerY)$}{
			$\mathit{fingerX} = \ispdomi{}{\mathit{fingerX}}$\;
		}
		\While{$\mathit{ordernr}(fingerY) < \mathit{ordernr}(fingerX)$}{
			$\mathit{fingerY} = \ispdomi{}{\mathit{fingerY}}$\;
		}
	}
	\KwRet{$\mathit{fingerX}$}
}
\end{algorithm}

Please be aware that \cref{algo:ImmediatePostDominance} is a general algorithm being applicable to all workflow nets. Thus, the nets may contain loops. For this reason, the algorithm requires a $\mathit{changed}$ variable, because in nets with loops, the (reverse) topological order is not fully satisfying. The $\mathit{changed}$ variable ensures that all information is finally correct. In our case of \emph{acyclic} nets, we could remove the $\mathit{changed}$ variable as the reverse topological order ensures the correct information after one iteration. Nevertheless, we decided to provide the more general algorithm at this place.

The computation of the post-dominance frontier of each node requires the immediate post-domi\-nance relation as stated in \cref{algo:PostDominanceFrontier}, \cref{al:pf:ipd}. The algorithm investigates each node (\cref{al:pf:en}) and already ensures that only nodes can be in the post-dominance frontier, which have at least two outputs (\cref{al:pf:div}). Each output $s$ of the current node $x$ is investigated and is assigned to a temporary variable $\mathit{runner}$ (\cref{al:pf:ru}). Confusingly for the first moment, the post-dominance frontier of $\mathit{runner}$ is updated. Taking $x = t10$ with its output $s = p17 = \mathit{runner}$ of \cref{fig:exampleAlgo} as an example: we know that $p17 \in \postset{t10}$ and $|\postset{t10}| \geq 2$. $\mathit{runner} = p17$ does not immediately post-dominate $x = t10$ but post-dominates itself. As a consequence, the ``post-dominance'' of $\mathit{runner} = s = p17$ ends with itself and, thus, $t10$ must be in $\mathit{runner}$'s ($p17$'s) post-dominance frontier (\cref{al:pf:up}). \cref{al:pf:next} updates $\mathit{runner}$ to its own immediate post-dominator $\ispdomi{}{\mathit{runner}}$. In the example, this is $t4$. $t4$ as $\mathit{runner}$ is not the immediate post-dominator of $x = t10$ again, but $t4$ immediately post-dominates $p17$ and, thus, an output of $x = t10$. For this reason, $x = t10$ also is within $t4$'s post-dominance frontier. Eventually, $\mathit{runner}$ gets $t5$. Since $t5$ is the immediate post-dominator of $x = t10$, the \emph{while}-loop in \cref{al:pf:lo} will finish. Be aware that no post-dominator of $t5$ will have $t10 = x$ within its post-dominance frontier \cite{Cooper}. \cref{algo:PostDominanceFrontier} finally terminates with the post-dominance frontiers of all nodes and returns them.

\begin{algorithm}[t]
\caption{Computing the post-dominance frontier for a workflow net $N=(P,T,F,i,o)$ following \cite{Cooper}.}
\label{algo:PostDominanceFrontier}
\SetKwProg{Fn}{Function}{}{end}
\Fn{postDomFrontier($N$)}{
	\tcp{Initialize}
	$\ispdomi{}{} \gets \text{computeImmediatePostDoms}(N)$\;\nllabel{al:pf:ipd}
	\For{$x \in P \cup T$}{
		$\pdomfront{x} \gets \emptyset$\;
	}
	\tcp{Compute}
	\For{$x \in P \cup T$}{\nllabel{al:pf:en}
		\If{$|\postset{x}| \geq 2$}{\nllabel{al:pf:div}
			\For{$s \in \postset{x}$}{
				$\mathit{runner} \gets s$\;\nllabel{al:pf:ru}
				\While{$\mathit{runner} \not = \ispdomi{}{x}$}{\nllabel{al:pf:lo}
				  $\pdomfront{\mathit{runner}} \gets \pdomfront{\mathit{runner}} \cup \Set{x}$\;\nllabel{al:pf:up}
					$\mathit{runner} \gets \ispdomi{}{\mathit{runner}}$\;\nllabel{al:pf:next}
				}
			}
		}
	}
	\KwRet{$\pdomfront{}$}
}
\end{algorithm}

As explained above, the \emph{iterated} post-dominance frontier $\itpdomfront{D}$ of a set of nodes $D$ is closely related to the diverging points of a node. For this reason, we are now able to compute the diverging points as stated in \cref{algo:DivergingPoints}. At first, it computes the post-dominance frontier of each node in \cref{al:dv:pdf}. Note that this frontier is independent of the marking $M_r$ to investigate. For this reason, it can be pre-computed. In \cref{al:dv:dpi}, the algorithm initializes the set $\divpoints{M_r}$ of diverging points of $M_r$ followed by the initialization of the mappings $\divinfoh{}{\markingset{M_r}}$ and $\reaches{}{\markingset{M_r}}$ in Lines~\ref{al:dv:di2} and \ref{al:dv:di3}. As the algorithm states in \cref{al:dv:di}, it computes $\divinfoS{x}{\markingset{M_r}}$ at the same time as the diverging points. \emph{Note at this place that we refrain to compute $\divinfo{o}{\markingset{M_r}}$ for each output node $o$ of a diverging point. Instead, the union of the information of all output nodes is computed for each diverging point as required in \cref{theorem:Reachability}.} %
In addition, we maintain the set $\reaches{x}{\markingset{M_r}}$, which stores to which diverging points and marked places of $M_r$ the current diverging point $x$ has direct paths (without any diverging point or place of $\markingset{M_r}$ in-between). \emph{Note that the resulting set $\reaches{x}{\markingset{M_r}}$ can contain other nodes than that of $\markingset{M_r}$.} %
From \cref{al:dv:itpdfs} to \cref{al:dv:itpdfe}, $\itpdomfront{\markingset{M_r}}$ is computed with its result stored in $\divpoints{M_r}$ and using a list $\mathit{list}$ that is continuously reduced and extended. The algorithm takes a node $\delta$  of the current node $x$'s post-dominance frontier and stores it as a diverging point (\cref{al:dv:st}). Then, it extends $\divinfoS{\delta}{\markingset{M_r}}$ of this diverging point by adding $x$ %
if $x \in M_r$ %
(since there is a path from $\delta$ to $x$) and by adding the set $\divinfoS{x}{\markingset{M_r}}$ (if $x \in \divpoints{M_r}$, then $\divinfoS{x}{\markingset{M_r}} \not = \emptyset$, otherwise it is empty). %
Furthermore, $\delta$ reaches $x$ and, therefore, stores it in $\reaches{\delta}{\markingset{M_r}}$ in \cref{al:dv:re}. %
Finally, $\delta$ is added at the end of $\mathit{list}$. For this reason, each diverging point is investigated again every time it newly appears within a post-dominance frontier of a node in $\mathit{list}$. For example, taking the marking $\multiset{p3,p8,p14,p17}$ in \cref{fig:exampleAlgo}, diverging point $t11$ will appear in $\pdomfront{p3}$ as well as in $\pdomfront{t10}$ (after $t10$ is added to $\mathit{list}$ as a diverging point of $p14$ and $p17$). For this reason, $t11$ will be investigated 3 times: once for $p17$ with $t10$, once for $p14$ with $t10$, and once for $p3$. Every time, the sets $\divinfoS{t11}{\markingset{M_r}}$ and $\reaches{t11}{\markingset{M_r}}$ will be extended accordingly. %
Since $\itpdomfront{\markingset{M_r}}$ is a superset of $\divpoints{M_r}$ as explained before, Lines~\ref{al:dv:rems}--\ref{al:dv:reme} remove those diverging points $\delta$, which do not have direct paths to at least two other diverging points or marked places of $M_r$ (i.\,e., $|\reaches{\delta}{\markingset{M_r}}| \leq 1$). %
Eventually, \cref{algo:DivergingPoints} returns the set of diverging points and $\divinfoh{}{\markingset{M_r}}$ as well as $\reaches{}{\markingset{M_r}}$ for each diverging point. %
Again: Computing the positions of $\phi$-functions is the same as computing the diverging points just in the reverse net. %
For this reason, it is not necessary to show that these algorithms finally result in the set of diverging points as this was already proven in \cite{DBLP:journals/toplas/CytronFRWZ91}.

\begin{algorithm}[t]
\caption{Computing the set of diverging points $\divpoints{M_r}$ for a workflow net $N=(P,T,F,i,o)$ and a given marking $M_r$, as well as the maps $\divinfoS{}{\markingset{M_r}}$ and $\reaches{}{\markingset{M_r}}$.}
\label{algo:DivergingPoints}
\SetKwProg{Fn}{Function}{}{end}
\Fn{computeDivergingPoints($N$, $M_r$)}{
	\tcp{Initialize}
	$\pdomfront{} \gets \text{postDomFrontier}(N)$\;\nllabel{al:dv:pdf}
	$\divpoints{M_r} \gets \emptyset$\;\nllabel{al:dv:dpi}
	\For{$x \in P \cup T$}{
		$\divinfoS{x}{\markingset{M_r}} \gets \emptyset$\;\nllabel{al:dv:di2}
		$\reaches{x}{\markingset{M_r}} \gets \emptyset$\;\nllabel{al:dv:di3}
	}
	$\mathit{list} \gets \markingset{M_r}$\;\nllabel{al:dv:itpdfs}
	\tcp{Compute. $\mathit{list}$ is ordered.}
	\While{$\mathit{list} \not = \emptyset$}{
		$x \gets \text{take first of } \mathit{list}$\;
		\For{$\delta \in \pdomfront{x}$}{
			$\divpoints{M_r} \gets \divpoints{M_r} \cup \Set{\delta}$\; \nllabel{al:dv:st}
			$\divinfoS{\delta}{\markingset{M_r}} \gets \divinfoS{\delta}{\markingset{M_r}} \cup \divinfoS{x}{\markingset{M_r}} \cup \big( \Set{x} \cap \markingset{M_r} \big)$\;\nllabel{al:dv:di}
			$\reaches{\delta}{\markingset{M_r}} \gets \reaches{\delta}{\markingset{M_r}} \cup \Set{x}$\;\nllabel{al:dv:re}
			$\mathit{list} \gets \mathit{list} \cup \Set{\delta}$\;
		}
	}	\nllabel{al:dv:itpdfe}
	\For{$\delta \in \divpoints{M_r}$}{\nllabel{al:dv:rems}
		\lIf{$|\reaches{\delta}{\markingset{M_r}}| \leq 1$}{
		  $\divpoints{M_r} \gets \divpoints{M_r} \setminus \Set{ \delta }$
		}
	}\nllabel{al:dv:reme}
	\KwRet{($\divpoints{M_r}$, $\divinfoS{}{\markingset{M_r}}$, $\reaches{}{\markingset{M_r}}$)}
}
\end{algorithm}

\cref{algo:ReachabilityAcyclic} summarizes the steps of checking (sub-marking) reachability of a given marking $M_r$ in a sound \AFW{} $N$. At first, (maximum) admissibility of $M_r$ is checked in \cref{al:rea:ad}. If $M_r$ is \emph{admissible} or \emph{maximum admissible}, then $M_r$ may be (sub-marking) reachable. Otherwise, if $M_r$ is \emph{not admissible}, $M_r$ is stated to be \emph{not reachable} and diagnostics could be provided. \cref{al:rea:dp} computes the diverging points of $M_r$  with \cref{algo:DivergingPoints}. This computation includes for each diverging point $\delta$ those marked places of $M_r$, $\divinfoS{\delta}{\markingset{M_r}}$, to which outputs of $\delta$ have paths. Finally, all diverging transitions $\delta$ are investigated in {Lines~\ref{al:rea:s1}--\ref{al:rea:s2}}. \cref{al:rea:ch} checks if $M_r$ is equal to $\divinfoS{\delta}{\markingset{M_r}}$. If this check holds, the marking is either \emph{coverable} (if $M_r$ is \emph{admissible}) or \emph{reachable} (if $M_r$ is maximum admissible). Otherwise, other diverging transitions are investigated until either there is one, for which the check holds, or the marking is stated as \emph{not reachable}. %
\cref{algo:ReachabilityAcyclic} does not require checking reachability by investigating the state space, computing occurrence sequences, etc. This explains why the computational complexity is just quadratic in the worst case:

\begin{algorithm}[t]
\caption{Checking (sub-marking) reachability of a given marking $M_r \in \Bag{P}$ for a sound \AFW{} $N=(P,T,F,i,o)$.}
\label{algo:ReachabilityAcyclic}
\SetKwProg{Fn}{Function}{}{end}
\Fn{isReachable($N$, $M_r$)}{
  $\mathit{admissibility},\mathit{missing},\mathit{conflict} \gets$ checkMaximumAdmissibility($N$, $M_r$)\; \nllabel{al:rea:ad}
	$\divpoints{M_r},\divinfoS{}{\markingset{M_r}},\reaches{}{\markingset{M_r}} \gets$ computeDivergingPoints($N$, $M_r$)\; \nllabel{al:rea:dp}
	\If{$\mathit{admissibility} \in \Set{$ admissible, maximum admissible $}$}{
		\For{$\delta \in (\divpoints{M_r} \cap T)$}{\nllabel{al:rea:s1}
			\If{$\markingset{M_r} = \divinfoS{\delta}{\markingset{M_r}}$}{ \nllabel{al:rea:ch}
				\uIf{$\mathit{admissibility} =$ admissible}{\nllabel{al:rea:ad1}
				  \tcp{Show expl. $\!\!\!\!\!\!\!$ using $\mathit{missing}$, $\divpoints{M_r}$, $\divinfoS{}{\markingset{M_r}}$, and $\reaches{}{\markingset{M_r}}$.}
					\KwRet{(coverable, $\mathit{missing}$, $\emptyset$, $\divpoints{M_r}$, $\divinfoS{}{\markingset{M_r}}$, $\reaches{}{\markingset{M_r}}$)}
				}\Else{
					\tcp{Show explanation using $\divpoints{M_r}$, $\divinfoS{}{\markingset{M_r}}$, and $\reaches{}{\markingset{M_r}}$.}
					\KwRet{(reachable, $\emptyset$, $\emptyset$, $\divpoints{M_r}$, $\divinfoS{}{\markingset{M_r}}$, $\reaches{}{\markingset{M_r}}$)}
				}\nllabel{al:rea:ad2}
			} \nllabel{al:rea:s3}
		}\nllabel{al:rea:s2}		
	}
	\tcp{Show diagnostics using $\mathit{conflict}$, $\mathit{missing}$, $\divpoints{M_r}$, $\divinfoS{}{\markingset{M_r}}$, and $\reaches{}{\markingset{M_r}}$.} \nllabel{al:rea:out2}
	\KwRet{(not reachable, $\mathit{missing}$, $\mathit{conflict}$, $\divpoints{M_r}$, $\divinfoS{}{\markingset{M_r}}$, $\reaches{}{\markingset{M_r}}$)}
}
\end{algorithm}

\begin{theoremit}[Computational Complexity]
\label{theorem:ComputationalComplexity}
Let $N=(P,T,F,i,o)$ be a sound \AFW{}  and $M_r$ a marking to check. Deciding reachability of $M_r$ can be computed in $O(|P|^2 + |T|^2)$. \defend
\end{theoremit}

\begin{proof}
\proofsize
Constructive proof by line-by-line investigation of \cref{algo:ReachabilityAcyclic}. %
\begin{description}
	\item[\cref{al:rea:ad}:] Checking admissibility and maximum admissibility can be achieved in quadratic computational time complexity, $O(|P|^2 + |T|^2)$, by \cref{theorem:ComputationalComplexityMax}. %
	\item[\cref{al:rea:dp}:] Computing the dominance frontier is feasible in quadratic time, $O(|P|^2 + |T|^2)$, by Cooper et al.\@ \cite{Cooper}, after building the immediate dominance relation in the same complexity. The computation of the diverging points for $M_r$ can also be achieved in quadratic time $O(|P|^2 + |T|^2)$ by Cytron et al.\@ \cite{DBLP:journals/toplas/CytronFRWZ91}. Thus, overall, this line takes $O(|P|^2 + |T|^2)$. 
	\item[Lines~\ref{al:rea:s1}--\ref{al:rea:s2}:] Lines~\ref{al:rea:ch}--\ref{al:rea:s3} are investigated $O(|T|)$ times by \cref{al:rea:s1}. Checking if $\markingset{M_r}$ is equal to $\divinfoS{\delta}{\markingset{M_r}}$ is constant in its complexity with an implementation requiring constant time for set operations (e.\,g., a BitSet in Java). The remaining Lines~\ref{al:rea:ad1}--\ref{al:rea:ad2} are constant, too. Thus, Lines~\ref{al:rea:s1}--\ref{al:rea:s2} can be achieved in $O(|T|)$ in general.
\end{description}

In summary, checking the (sub-marking) reachability of a given marking is dominated by the term $O(|P|^2 + |T|^2)$. $\checkmark$

\emph{Termination:} All steps in \cref{algo:ReachabilityAcyclic} are either known to terminate or are performed over finite sets (as over the set of diverging points, $M_r$, etc.). For this reason, \cref{algo:ReachabilityAcyclic} terminates. \proofend
\end{proof}

\subsection{Diagnostics}
\label{subsec:DiagnosticsAll}

Each output of \cref{algo:ReachabilityAcyclic} provides information that can be used as diagnostics or explanations. Whereas the information about $\mathit{missing}$ and $\mathit{conflicting}$ places was already discussed in \cref{subsec:AdmissibilityOutput}, diagnostics with knowledge about diverging points will be discussed in the following.

Let $N$ be a net with $M_r$ is a marking to check. Computing the diverging points of $M_r$ with \cref{algo:DivergingPoints} will lead to the set $\divpoints{M_r}$ and the mappings $\divinfoS{}{\markingset{M_r}}$ and $\reaches{}{\markingset{M_r}}$. For the moment, let us consider a single diverging point $\delta \in \divpoints{M_r}$. Independently of whether $\delta \in P$ or $\delta \in T$, $\reaches{\delta}{\markingset{M_r}}$ provides all information to which nodes of $\divpoints{M_r} \cup \markingset{M_r}$ the diverging point $\delta$ has direct paths. For example, for the net of \cref{fig:examplep3p8p14p17dia}, the marking $\multiset{p3,p8,p14,17}$ shall be investigated. This marking is maximum admissible as explained before. \Cref{algo:DivergingPoints} will provide $\reaches{t10}{\Set{p3,p8,p14,17}} = \Set{p14,p17}$. By \cref{def:DivergingPoint}, there must be two disjoint paths from $t10$ to $p14$ and from $t10$ to $p17$ (except $t10$), which can be algorithmically detected. Thus, the diverging point $t10$ with the disjoint paths to $p14$ and to $p17$ can be highlighted to show that $t10$ as a diverging transition may ``cause'' concurrency of $p14$ and $p17$. \Cref{fig:examplep3p8p14p17dia} shows this highlighting in orange. Similarly, \cref{algo:DivergingPoints} will provide $\reaches{t11}{\Set{p3,p8,p14,17}} = \Set{p3,t10}$. In this case, on any path from $t11$ to $p14$ and to $p17$ lies $t10$. For this reason, $t11$ may indirectly ``cause'' concurrency of $p14$ and $p17$, which can be used in an explanation why sub-marking $\multiset{p3,p14,p17}$ is reachable. Again, the paths from $t11$ to $p3$ and to $t10$ can be highlighted as provided in \cref{fig:examplep3p8p14p17dia} in blue. In the case, there is a diverging place, such as $p16$ in the figure, the information $\reaches{p16}{\Set{p3,p8,p14,17}} = \Set{t11,t12}$ is not fine grained enough to directly select one output for highlighting. However, by \cref{def:Dominance} of post-dominance frontiers, at least two nodes in $\reaches{p16}{\Set{p3,p8,p14,17}}$ must post-dominate at least two separate outputs of $p16$; otherwise, $p16$ would not be within the iterated post-dominance frontier and, thus, not be a diverging point of $\Set{p3,p8,p14,17}$. In the current case, $t11$ and $t12$ as outputs of $p16$ post-dominate themselves. Thus, it is feasible to select that output, which is post-dominated by that node with the greatest subset of $M_r$ in $\divinfoS{t11}{\Set{p3,p8,p14,17}}$ or $\divinfoS{t12}{\Set{p3,p8,p14,17}}$. Both $t11$ and $t12$ have the same information, $\divinfoS{t11}{\Set{p3,p8,p14,17}} = \divinfoS{t12}{\Set{p3,p8,p14,17}}$, therefore, selecting one is random. In the figure, $t11$ was taken in violet and blue. As the reader can step-wisely investigate, the diagnostics provided in \cref{fig:examplep3p8p14p17dia} can be derived out of \cref{algo:DivergingPoints}. Furthermore, these diagnostics can further be used to force a simulation of the reachable marking.

\begin{figure}[t]
	\centering
		\includegraphics[width=0.80\textwidth]{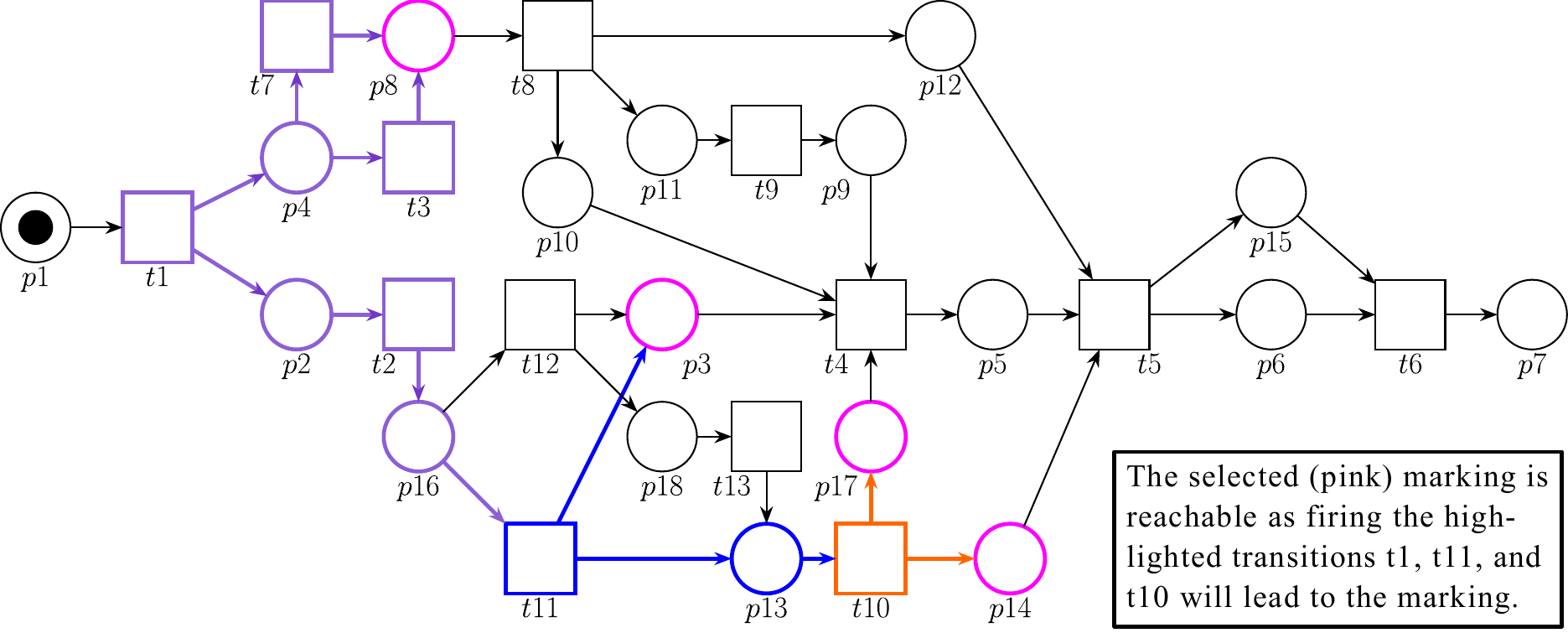}
		\vspace{-0.35cm}
	\caption{The net of \cref{fig:example} investigating the maximum admissible marking $\multiset{p3,p8,p14,p17}$ (colored in pink). Parts of the diverging points of this marking can be used to explain why this marking is reachable (cf. \cref{fig:examplep3p8p14p17}).}
	\label{fig:examplep3p8p14p17dia}
\end{figure}

Until now, the discussion focused on (sub-marking) reachable markings. However, the information about diverging points can also be used in the case of a non-admissible marking $M_r$. In \cref{subsec:AdmissibilityOutput}, it was discussed that two places $x$ and $y$ of $M_r$ can be in conflict ($\notconcurrent{x}{y}$) in three ways: (1) $x$ has a path to $y$, (2) $y$ has a path to $x$, or (3) there is no path between $x$ and $y$. As cases~(1) and (2) were already discussed, how to derive information in case (3) remained open. If $x$ and $y$ do not have any path between them, then there must be a diverging point of $x$ and $y$ on any path from the source $i$ of the workflow net to $x$ and $y$. Such diverging point cannot be a transition as, otherwise, $\concurrent{x}{y}$ \cite{DBLP:conf/apn/PrinzKB24}. As a consequence, in case~(3), we can use $\reaches{}{\markingset{M_r}}$ and $\divpoints{\Set{x,y}}$ to traverse the net in reverse direction starting from $x$ and $y$ until we reach a diverging place of $x$ and $y$. As explained above, the paths from this diverging place to $x$ and $y$ can be highlighted. Such an example was already illustrated in \cref{fig:example-divpoints-2-p3p5p7} with marking $\multiset{p3,p5,p6}$ to investigate. In this example, it holds that $\divpoints{\multiset{p3,p5,p7}} = \Set{t1,p4}$ with $\reaches{p4}{\Set{p3,p5,p7}} = \Set{p5,p7}$ and $\reaches{t1}{\Set{p3,p5,p7}} = \Set{p3,p4}$ (note that, of course, $t1$ will \emph{not} lead to a reachable marking $\multiset{p3,p5,p7}$ by \cref{theorem:Reachability} since this marking is \emph{not} admissible). Thus, $p4$ is a fitting diverging place of $p5$ and $p7$ so that the paths from $p4$ to $p5$ and from $p4$ to $p7$ can be highlighted as done in \cref{fig:example-divpoints-2-p3p5p7} in orange and as discussed before. 

Regarding the provided diagnostics and efficiency, the here introduced reachability approach has a promising quality of diagnostic information that can be used as discussed and as illustrated within this paper. Its efficiency and completeness is well suited for the given class of sound acyclic free-choice workflow nets. For its special use case, it seems to be the best option available and could be implemented in tools and for further analysis while providing qualitative high feedback.

\subsection{Extended Free-Choiceness}
\label{subsec:ExtendedFC}

Some sound \AFWs{} are \emph{extended} free-choice. Although this paper is based on \emph{simple} free-choiceness, extended free-choiceness can be handled with the here presented approach: %
Sound acyclic \emph{extended} free-choice workflow nets can be transformed into \emph{simple} free-choiceness with the linear time transformation by Murata \cite{Best1987,DBLP:journals/pieee/Murata89} \emph{before} analysis (illustrated in \cref{fig:ExtendedToSimpleFC}). This transformation inserts a new transition and a new place by eliminating a couple of flows, but by retaining the behavior and size of the model (in terms of its complexity class). After transformation, the model allows for some additional reachable markings containing the newly inserted places ($p$ in \cref{fig:ExtendedToSimpleFC}). However, these additional markings will not be investigated for (sub-marking) reachability as the new places are not present in the original net. Therefore, the transformation will not influence reachability with the presented approach. As a consequence, sound extended \AFWs{} can also be investigated with the transformation of Murata \cite{DBLP:journals/pieee/Murata89} without changing the algorithms and the overall computational complexity.

\begin{figure}[tb]
	\begin{minipage}[t]{0.48\textwidth}
		\begin{center}
			\includegraphics[width=0.28\textwidth]{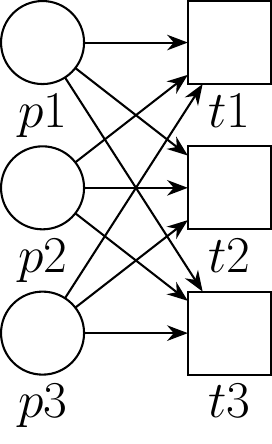}\\
			(a)
		\end{center}
	\end{minipage} \hfill %
	\begin{minipage}[t]{0.48\textwidth}
		\begin{center}
			\includegraphics[width=0.6\textwidth]{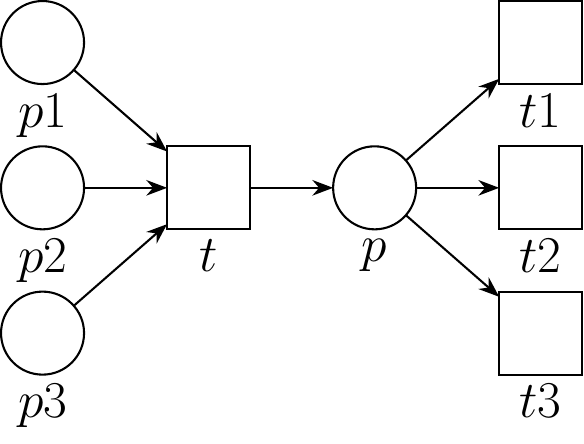}\\
			(b)
		\end{center}
	\end{minipage}	
	\caption{Transforming an extended free-choice structure (a) into a simple free-choice structure (b) by inserting a new transition and place.}
	\label{fig:ExtendedToSimpleFC}
\end{figure}

\section{Conclusion}
\label{sec:Conclusion}

Reachability and the covering problem are central decision problems in Petri nets theory. This paper showed for both problems that they can be decided in a quadratic computational time complexity for the class of sound acyclic free-choice workflow nets (more specifically, in $O(|P|^2 + |T|^2)$). Furthermore, this paper explained on algorithms how the pure decisions can be extended with diagnostics allowing for giving explanations on why a given marking is or is not (sub-)marking reachable. In doing this, the concept of admissibility was introduced, in which all marked places of a marking are pairwise concurrent. In a maximum admissible marking, adding further tokens to places would destroy the marking's admissibility. All reachable markings are maximum admissible, whereas all coverable markings are admissible. Non-admissible markings contain conflicting tokens on places, which can be used as diagnostics. Since admissibility and maximum admissibility are just necessary conditions of the covering and reachability problems, respectively, this paper introduced diverging transitions as sufficient condition. Diverging transitions are transitions with disjoint paths to marked places in a given marking. Those disjoint paths manifest in a subgraph of the net representing a possible execution. Thus, the existence of a diverging transition for an admissible marking implies the reachability or coverage of this marking. The implied subgraph by such transitions can be used as diagnostics to explain why a marking is reachable. If such a transition does not exist, then a similar subgraph can be used as diagnostics to explain why a given marking is not reachable. Detailed algorithms allow for a straight-forward application of the approach to existing Petri net tools.

The Petri nets and business process management communities benefit from new concepts (admissibility and diverging transitions) and from new computationally efficient algorithms for deciding the covering and reachability problems, especially for process-like and industrial related nets. If required, the algorithms' outputs can be used for providing diagnostics to users such as system/business analysts, students, etc. For example, business analysts could request their process models whether undesired states of their models could happen, or not --- to improve quality, to check compliance, and to avoid cost-intensive failures. Since the reachability of markings is central to many other decision problems of Petri nets, the here presented approach enables to speed up other algorithms. 

In the future, the concepts of admissibility and maximum admissibility should be generalized for the class of proper free-choice nets with a home cluster. Since diverging transitions are a structural concept, possible further applications for them should be identified. Furthermore, the method of loop decomposition \cite{DBLP:conf/apn/PrinzKB24,DBLP:journals/is/PrinzCH25,DBLP:conf/bpm/PrinzCH22} should be introduced for free-choice nets leading to an unfolding of loops increasing a model just quadratically in the worst case. This may close the gap between the application of the provided algorithms from acyclic to cyclic nets. For industrial settings, we want to provide an implementation of our approach, which allows for checking reachability for process models (e.\,g., in BPMN \cite{OMG2011}) by simply selecting some nodes of a model. %
Finally, it is of interest whether the new concepts introduced in this paper may be applicable to solve other problems in Petri nets.

\bibliographystyle{fundam}
\bibliography{bibliography}


\end{document}